\documentclass{article}

\usepackage{PRIMEarxiv}
\usepackage[utf8]{inputenc}
\usepackage[T1]{fontenc}
\usepackage{hyperref}
\usepackage{url}
\usepackage{booktabs}
\usepackage{amsfonts}
\usepackage{nicefrac}
\usepackage{microtype}
\usepackage{lipsum}
\usepackage{fancyhdr}
\usepackage{graphicx}
\usepackage{tabularx}
\usepackage{enumitem}
\usepackage{amsmath,amssymb,amsfonts}
\usepackage{amsthm}
\usepackage{algorithm}
\usepackage{algpseudocode}
\usepackage{caption}
\usepackage{thmtools}
\usepackage{thm-restate}
\graphicspath{{media/}}

\newtheorem{theorem}{Theorem}
\newtheorem{proposition}[theorem]{Proposition}
\newtheorem{lemma}[theorem]{Lemma}
\newtheorem{definition}{Definition}
\setlist[enumerate,1]{label=(\roman*),font=\normalfont}
\title{Geometric Structure and Polynomial-time Algorithm of Game Equilibria
}

\author{
  Hongbo Sun \thanks{For discussing technical details, please contact this author.} \\
  Shenzhen International\\ Graduate School \\
  Tsinghua University \\
  Shenzhen\\
  \texttt{shb20@tsinghua.org.cn} \\
  \And
  Chongkun Xia \\
  School of Advanced Manufacturing \\
  Sun Yat-sen University \\
  Shenzhen\\
  \texttt{xiachk5@mail.sysu.edu.cn} \\
  \And
  Junbo Tan \\
  Shenzhen International\\ Graduate School \\
  Tsinghua University \\
  Shenzhen\\
  \texttt{tjblql@sz.tsinghua.edu.cn} \\
  \And
  Bo Yuan \\
  Research Institute of\\ Tsinghua University in Shenzhen \\
  Tsinghua University \\
  Shenzhen\\
  \texttt{boyuan@ieee.org} \\
  \And
  Xueqian Wang \\
  Shenzhen International Graduate School \\
  Tsinghua University \\
  Shenzhen\\
  \texttt{wang.xq@sz.tsinghua.edu.cn} \\
  \And
  Bin Liang \\
  Department of Automation \\
  Tsinghua University \\
  Beijing\\
  \texttt{bliang@tsinghua.edu.cn} \\
}

\begin{document}
\maketitle

\begin{abstract}
    Whether a PTAS (polynomial-time approximation scheme) exists for game equilibria has been an open question, and its absence has indications and consequences in three fields:
    the practicality of methods in algorithmic game theory, non-stationarity and curse of multiagency in MARL (multi-agent reinforcement learning), and the tractability of PPAD in computational complexity theory.
    In this paper, we formalize the game equilibrium problem as an optimization problem that splits into two subproblems with respect to policy and value function, which are solved respectively by interior point method and dynamic programming.
    Combining these two parts, we obtain an FPTAS (fully PTAS) for the weak approximation (approximating to an $\epsilon$-equilibrium) of any perfect equilibrium of any dynamic game, implying PPAD=FP since the weak approximation problem is PPAD-complete.
    In addition, we introduce a geometric object called equilibrium bundle,
    regarding which, first, perfect equilibria of dynamic games are formalized as zero points of its canonical section,
    second, the hybrid iteration of dynamic programming and interior point method is formalized as a line search on it,
    third, it derives the existence and oddness theorems as an extension of those of Nash equilibria.
    In experiment, the line search process is animated, and the method is tested on 2000 randomly generated dynamic games where it converges to a perfect equilibrium in every single case.
\end{abstract}

\keywords{game theory \and equilibrium \and dynamic programming \and interior point method \and polynomial-time approximation scheme}
\textbf{MSC codes:} 90C39, 90C51, 91A15

\section{Introduction}
\subsection{The computational problem of game equilibria}
Whether game equilibria can be efficiently solved has been an open question since Nash equilibrium\cite{nash_orginal} is proposed.
The absence of such a polynomial-time algorithm has indications and consequences in three different fields.
\begin{itemize}
    \item In algorithmic game theory, the most widely used algorithms for approximating equilibria are no-regret\cite{noregret,cfr}, self-play\cite{fp,fsp}, and their variants.
          \begin{itemize}
              \item Neither of them converges to Nash equilibria in general static games.
                    No-regret converges to coarse correlated equilibria, and self-play converges to strict Nash equilibria in static games that satisfy the fictitious play property.
              \item When applied to dynamic games, neither of them guarantees the approximated equilibria to be perfect, namely, being a Nash equilibrium at every stage of the game, and thus non-perfect equilibria are not optimal in dynamic games.
          \end{itemize}
    \item In MARL, there are problems known as non-stationarity and curse of multiagency\cite{mrl_problem}, and these two problems are related to problems in algorithmic game theory, such that there is currently no algorithm that can converge to Nash equilibria or perfect equilibria in polynomial time for general games.
          \begin{itemize}
              \item Non-stationarity means that the policies are hard to converge when reinforcement learning agents simultaneously maximize their utilities.
              \item Curse of multiagency means that the computation needed for the policies of all agents to achieve optimal is exponential to the number of agents.
          \end{itemize}
    \item In computational complexity theory, there are the following results about approximating Nash equilibria.
          \begin{itemize}
              \item There are two different approximations: weak approximation and strong approximation\cite{different_appro}.
                    \begin{itemize}
                        \item Weak approximation is the approximation to an $\epsilon$-equilibrium, namely, a policy profile such that every player is at most $\epsilon$ away from its maximum utility.
                        \item Strong approximation is the approximation to an $\epsilon$-neighborhood of an exact equilibrium, namely, a policy profile whose distance from an exact equilibrium is less than $\epsilon$.
                        \item A weak approximation of an exact equilibrium could be far from the exact equilibrium.
                    \end{itemize}
              \item For static games with any number of players, computing a weak approximation of Nash equilibria is PPAD-complete\cite{nash_weak_complex1}, if $\epsilon$ is inversely proportional to a polynomial in the game size\cite{nash_weak_complex0},
                    indicating that an FPTAS does not exist for weak approximation of Nash equilibria unless PPAD=FP.
              \item For static games with any number of players, computing a weak approximation of Nash equilibria for fixed $\epsilon$ is also PPAD-complete\cite{nash_weak_complex2},
                    indicating that an PTAS does not exist for weak approximation of Nash equilibria unless PPAD=FP.
              \item For static games with two players, computing an exact Nash equilibrium is PPAD-complete\cite{nash_strong_complex0}.
              \item For static games with three or more players, computing an exact Nash equilibrium and computing a strong approximation of Nash equilibria are both FIXP-complete\cite{nash_strong_complex1}.
          \end{itemize}
\end{itemize}

In this paper, we deal with fully observable dynamic games, and all the results of dynamic games hold for static games as single-state degenerations.
Our results eventually lead to an FPTAS for weak approximation of any perfect equilibrium of any dynamic game, implying PPAD=FP.

PTAS, FPTAS, PPAD, FP are all general concepts in computational complexity theory.
A PTAS (polynomial-time approximation scheme) for a class of problems is an algorithm that approximates an answer to precision $\epsilon$ in polynomial time with respect to the problem size.
An FPTAS (fully polynomial-time approximation scheme) is a PTAS whose running time is not only a polynomial with respect to the problem size, but also a polynomial with respect to $1/\epsilon$.
FP is a complexity class that consists of function problems that can be solved in polynomial time, corresponding to the well-known complexity class P, which consists of decision problems that can be solved in polynomial time.
Decision problems are the computational problems for which the answer is simply yes or no, while function problems are the computational problems for which the answer can be anything.
PPAD is a complexity class defined as all the function problems that reduce to an {\it End-Of-The-Line} problem in polynomial time\cite{ppad},
and {\it End-Of-The-Line} is pretty convincingly intractable in polynomial time, making PPAD believed to contain hard problems, that is, it is believed to be unlikely that PPAD=FP.

In {\it End-Of-The-Line}, there is a directed graph $DG$, and a polynomial-time computable function $f$ given by a boolean circuit.
$DG$ consists of vertices that are connected by arrows, where each vertex has at most one predecessor and at most one successor.
Each vertex is encoded in $n$ bits, and $f$ takes the bits of a vertex as input to output its predecessor and successor, either of which may be none.
A vertex is called an unbalanced vertex if exactly one of its predecessor and successor is none.
    {\it End-Of-The-Line} is the problem that: given an unbalanced vertex, find another unbalanced vertex.
In $DG$, every vertex is either on a chain or isolated, and the problem seems to let us follow a potentially exponentially long chain from a given source to a sink,
making {\it End-Of-The-Line} pretty convincingly intractable in polynomial time.

PPAD stands for Polynomial Parity Arguments on Directed graphs, where the parity argument refers to that given an unbalanced vertex, there exist an odd number of other unbalanced vertices.
The oddness theorem\cite{oddness2} of Nash equilibria stating that there are an odd number of Nash equilibria for almost all\footnote{Almost all means other cases form a null set in all the cases, that is, the probability of encountering them is zero.} static games is also a parity argument.
The Sperner's lemma in combinatorics, which is equivalent to Brouwer's fixed point theorem that is used to prove the existence\cite{nash_orginal2} of Nash equilibria, is another parity argument.
All these three parity arguments are based on the fact that the solutions are connected in pairs.

\subsection{Notation}

Notations in this paper generally follow Einstein summation convention since the operations involved are all tensor operations.
In mathematics and physics, Einstein summation convention is a widely used notation to achieve brevity in denoting tensor operations.
In computer science, Einstein summation is also one of the basic functions implemented in tensor computing platforms such as NumPy and PyTorch.
In this paper, we have attached an experimental code implementing our algorithm, which is based on the {\it einsum} function in NumPy and PyTorch,
and thus a good illustration of Einstein summation convention can be found in their documentation.

First, a tensor is denoted as a symbol with superscripts and subscripts such as $\mu_a^i$, where the superscripts and subscripts represent the indices of the tensor, and the range of an index such as $a\in \mathcal{A}$ is implicitly inferred from context.
A tensor can be denoted using different index symbols, such as $\mu_a^i$ and $\mu_{a''}^k$ represent the same tensor.
In particular, $\mathbf{1}_s$ represents the tensor whose elements are all $1$, and $I_{xy}$, $I^{xy}$, or $I_x^y$ is the identity matrix, which is a second-order tensor with two indices.

Second, Einstein summation in this paper is an operation that first takes element-wise product and then sums over vanishing indices, which follows the {\it einsum} function in NumPy and PyTorch, instead of the classic Einstein summation.
For example, the Einstein summation of $Y_a^s$ and $Z_a^{si}$ is denoted as $Y_a^s Z_a^{si}$, and the indices it sums over is inferred from context, such as the equation
$$W_s^i+X^i=Y_a^s Z_a^{si}:=\sum_{a\in \mathcal{A}} Y_a^s Z_a^{si},$$
where indices $s$ and $i$ appear in other additives of the equation, while index $a$ vanishes.
Thus, the Einstein summation $Y_a^s Z_a^{si}$ is the summation over index $a$, where the range of $a$ is implicitly inferred from outer context.
In particular, when an Einstein summation appears in the objective of an optimization problem where there is no equation, such as $w^s \pi_a^{si} r_a^{si}$,
the summation over all the indices is implied since the objective is a scalar by definition such that all the indices vanish.
In addition, element-wise product with no summation over any index is explicitly denoted using Hadamard product, such as $\pi_a^{si}\circ r_a^{si}$, though it can be implicitly inferred from context.

Third, there are a few unconventional notations to further simplify the expressions, which make sense after the problem definition in the next subsection.
For policy $\pi_a^{si}$, let
\begin{equation*}
    \pi_A^s:=\left(\prod_{k\in N}\pi_{a_k}^{sk}\right)_A^s,
    \pi_{Aa}^{si-}:=\left(I_{a_i a}\prod_{k\in N-\{i\} }\pi_{a_k}^{sk}\right)_{Aa}^{si},
    \pi_{Aaa'}^{sij-}:=\left([i\neq j]\cdot I_{a_i a}I_{a_j a'}\prod_{k\in N-\left\{i,j\right\} }\pi_{a_k}^{sk}\right)_{Aaa'}^{sij},
\end{equation*}
where $A$ represents $(a_i)_{i\in N}$ and $[\cdot]$ is the Iverson bracket that outputs $1$ if the input is true and outputs $0$ otherwise.
Then, let $\pi_A^sU_{\pi A}^{si}:=(\pi_A^sU_{\pi A}^{si})^{si}$, $\pi_{Aa}^{si-}U_{\pi A}^{si}:=(\pi_{Aa}^{si-}U_{\pi A}^{si})_a^{si}$, and $\pi_{Aaa'}^{sij-}U_{\pi A}^{si}:=(\pi_{Aaa'}^{sij-}U_{\pi A}^{si})_{aa'}^{sij}$.
Let $T_{\pi s'}^s:=\pi_a^s T_{s'a}^s$ and $u_\pi^s:=\pi_a^s u_a^s$.
Let $\max_a$ represent the maximum with respect to index $a$ for every other indices.
Finally, in dealing with static games, we may drop the index $s$ when one state is referred and put it on when all states are referred, such as $\pi_a^i$ represents $\pi_a^{si}(x)$ for some given state $x$.

\subsection{Problem definition}

\begin{definition}[Dynamic game]
    \label{dynamic_game}
    A dynamic game $\Gamma$ is a tuple $(N,\mathcal{S},\mathcal{A},T,u,\gamma)$, where $N$ is a set of players, $\mathcal{S}$ is a state space, $\mathcal{A}$ is an action space,
    $T:\mathcal{S} \times \prod_{i\in N} \mathcal{A} \to \Delta(\mathcal{S}) $ is a transition function, $\Delta(\mathcal{S})$ is a space of probability distributions on $\mathcal{S}$,
    $u:\mathcal{S} \times \prod_{i\in N} \mathcal{A} \to \prod_{i\in N} \mathbb{R}$ is an utility function, $\mathbb{R}$ is the set of real numbers, $\gamma \in [0,1)$ is a discount factor.
    A static game is a dynamic game with only one element in its state space $\mathcal{S}$.
\end{definition}
\begin{definition}[Perfect equilibrium]
    \label{perfect_equil}
    A perfect equilibrium of dynamic game $\Gamma$ is a policy $\pi:\mathcal{S} \to \prod_{i\in N} \Delta(\mathcal{A} )$ such that $V_s^i=\max_a\pi_{Aa}^{si-}(u_A^{si}+\gamma T_{s'A}^s V_{s'}^i)$,
    where $V:\mathcal{S} \to \prod_{i\in N} \mathbb{R}$ is the unique value function that satisfies $V_s^i=\pi_A^s(u_A^{si}+\gamma T_{s'A}^s V_{s'}^i)$.
    A Nash equilibrium is a perfect equilibrium of a static game.
\end{definition}

Dynamic game and perfect equilibrium inherit the major characteristics of stochastic game and Markov perfect equilibrium\cite{stocgame,mpe} in existing research, and they are defined in the most intuitive and general way.
Then, static games and Nash equilibria are the single-state degenerations of dynamic games and perfect equilibria.

Dynamic game generally follows the state-space representation.
First, on a given state $s\in\mathcal{S}$, a set $N$ of players choose their actions in the action space $\mathcal{A}$ simultaneously.
Second, according to the action combination $(a_i)_{i\in N}$, the utility function $u$ specifies the utility of each player, and the transition function $T$ specifies the next state $s'\in\mathcal{S}$ of the game.
Third, the process is repeated, and the goal of each player $i\in N$ is to maximize its cumulative utility discounted by factor $\gamma$ on every initial state.
We denote the utility function as tensor $u_A^{si}$ and the transition function as tensor $T_{s'A}^s$, where $A:=(a_i)_{i\in N}$ is the index of action combination.

The policy $\pi$ specifies a probability distribution $\Delta(\mathcal{A})$ of actions for each player $i\in N$ on every state $s\in\mathcal{S}$.
The value function $V$ measures the value of each state $s\in\mathcal{S}$ to each player $i\in N$.
Denote $\mathcal{P}=\{\pi:\mathcal{S} \to \prod_{i\in N} \Delta(\mathcal{A} )\}$ as the policy space, and $\mathcal{V}=\{V:\mathcal{S} \to \prod_{i\in N} \mathbb{R}\}$ as the value function space.
We denote the policy as tensor $\pi_a^{si}$ and the value function as tensor $V_s^i$.
Note that for any policy $\pi_a^{si}$, there is a unique value function satisfying $V_s^i=\pi_A^s(u_A^{si}+\gamma T_{s'A}^s V_{s'}^i)$, which we denote as $V_{\pi s}^i$,
because $u_\pi^s=(I_{s'}^s-\gamma T_{\pi s'}^s)V_{\pi s'}$ and $I_{s'}^s-\gamma T_{\pi s'}^s$ is invertible by Lemma~\ref{stoc_mat_lemm} (ii), where $T_{\pi s'}^s:=\pi_a^s T_{s'a}^s$ and $u_\pi^s:=\pi_a^s u_a^s$.
In addition, we call $V_{\pi s}^i$ the value function of policy $\pi_a^{si}$ because, intuitively, for any state $x\in\mathcal{S}$, $V_{\pi s}^i(x)$ is the discounted cumulative utility generated by $\pi_a^{si}$ on initial state $x$.

A perfect equilibrium is a policy where all the players simultaneously earn their maximum discounted cumulative utilities on every initial state, which involves two equations.
First, $V_s^i=\pi_A^s(u_A^{si}+\gamma T_{s'A}^s V_{s'}^i)$ is to ensure that value function $V_s^i$ is the discounted cumulative utility generated by policy $\pi_a^{si}$ on every initial state.
Second, $V_s^i=\max_a\pi_{Aa}^{si-}(u_A^{si}+\gamma T_{s'A}^s V_{s'}^i)$ is to ensure that $V_s^i$ is the maximum discounted cumulative utilities of each player on every initial state.
A weak approximation of a perfect equilibrium is an $\epsilon$-perfect equilibrium, which is a policy that every player is at most $\epsilon$ away from its maximum utility.
There are many equivalent definitions of $\epsilon$-Nash equilibrium, but we use equation \eqref{eps_equil} as the definition of $\epsilon$-perfect equilibrium, and its single-state degeneration can be shown equivalent to existing $\epsilon$-Nash equilibrium definitions.

\subsection{Major results}

In this paper, we introduce an equation called unbiased KKT conditions and a geometric object called equilibrium bundle, along with unbiased barrier problem and Brouwer function.
The major results are centered around the unbiased KKT conditions and the equilibrium bundle, which are respectively introduced as we go through the following sections, but we put their definitions here only to describe our major results.

\begin{restatable}[Unbiased KKT conditions of dynamic games]{definition}{ukktdy}
    Unbiased KKT conditions are given by the following simultaneous equations,
    where $U_{\pi A}^{si}=u_A^{si}+\gamma T_{s'A}^s V_{\pi s'}^i$ and $V_{\pi s}^i$ is the unique value function satisfying $V_s^i=\pi_A^s(u_A^{si}+\gamma T_{s'A}^s V_{s'}^i)$.
    \begin{subequations}
        \label{ukkt_dy_equ}
        \begin{equation}
            \label{ukkt_dy_equ1}
            \begin{bmatrix}
                \hat{\pi}_a^{si}\circ r_a^{si}-\mu_a^{si}   \\
                r_a^{si}-v_s^i+\pi_{Aa}^{si-}U_{\pi A}^{si} \\
                \mathbf{1}_a\hat{\pi}_a^{si}-\mathbf{1}^{si}
            \end{bmatrix}=0
        \end{equation}
        \begin{equation}
            \pi_a^{si}=\hat{\pi}_a^{si}
        \end{equation}
    \end{subequations}
\end{restatable}

\begin{restatable}[Equilibrium bundle of dynamic games]{definition}{equilbundl}
    An equilibrium bundle is the tuple $(E, \mathcal{P}, \alpha:E\to \mathcal{P})$ given by the following equations,
    where $U_{\pi A}^{si}=u_A^{si}+\gamma T_{s'A}^s V_{\pi s'}^i$ and $V_{\pi s}^i$ is the unique value function satisfying $V_s^i=\pi_A^s(u_A^{si}+\gamma T_{s'A}^s V_{s'}^i)$.
    \begin{equation}
        \label{equil_bund_dy_equ}
        \begin{aligned}
             & E=\bigcup _{\pi_a^{si} \in \mathcal{P}} \{\pi_a^{si}\} \times B(\pi_a^{si})                                                 \\
             & B(\pi_a^{si})=\left\{v_s^i\circ\pi_a^{si}+\bar{\mu}_a^{si}(\pi_a^{si})|v_s^i\geq 0\right\}                                  \\
             & \bar{\mu}_a^{si}(\pi_a^{si})=\pi_a^{si} \circ \left(\max_a \pi_{Aa}^{si-}U_{\pi A}^{si}-\pi_{Aa}^{si-}U_{\pi A}^{si}\right) \\
             & \alpha\left((\pi_a^{si},\mu_a^{si})\right) =\pi_a^{si}                                                                      \\
        \end{aligned}
    \end{equation}
    The map $\bar{\mu}_a^{si}:\mathcal{P}\to \{\mu_a^{si}|\min_a^{si}\mu_a^{si}=\mathbf{0}^{si}\}$ is called the canonical section of the equilibrium bundle.
\end{restatable}

The equilibrium bundle is a geometric object called fiber bundle\footnote{We only use the basic idea that the total space is the disjoint union of the fibers over the base space to formalize a geometric structure, no deeper theories are involved.} in differential geometry,
its base space is policy space $\mathcal{P}$, its total space is $E$ consisting of pairs $(\pi_a^{si},\mu_a^{si})$, its fiber over $\pi_a^{si}\in \mathcal{P}$ is $\alpha^{-1}(\pi_a^{si})=\{\pi_a^{si}\} \times B(\pi_a^{si})$,
and canonical section $\bar{\mu}_a^{si}$ is its section that maps each $\pi_a^{si}\in \mathcal{P}$ to the least element in $B(\pi_a^{si})$.
Conventionally, we can also call $E$ as the equilibrium bundle, and call $B(\pi_a^{si})$ as its fiber over $\pi_a^{si}$.
The major results are several facts about the unbiased KKT conditions and the equilibrium bundle, which we show as we go through the following sections.

\begin{itemize}
    \item There are many equivalent characterizations of perfect equilibria of dynamic games centered around the unbiased KKT conditions.
          \begin{itemize}
              \item For every barrier parameter, unbiased KKT conditions derive a Brouwer function that maps a primal policy to a dual policy, such that the Brouwer's fixed point theorem applies, leading to the existence theorem of perfect equilibrium as an extension of that of Nash equilibrium\cite{nash_orginal2}.
                    \begin{itemize}
                        \item There is an unbiased barrier problem depicting the approximation of the fixed points of the Brouwer function,
                              such that the fixed points are its global optimal points, or equivalently its zero points.
                    \end{itemize}
              \item Unbiased KKT conditions derive a polynomial function that maps a policy to a barrier parameter, such that the Newton-Puiseux theorem applies, leading to the oddness theorem of perfect equilibrium as an extension of that of Nash equilibrium\cite{oddness2}.
              \item The solution space of unbiased KKT conditions can be structured as an equilibrium bundle, such that perfect equilibria are the zero points of the canonical section of the equilibrium bundle, where the canonical section depicts the global distribution of perfect equilibria in policy space.
          \end{itemize}
    \item There is a hybrid iteration of dynamic programming and interior point method formalized as a line search method on the equilibrium bundle.
          \begin{itemize}
              \item The method consists of two levels of iteration: updating onto the equilibrium bundle by alternating the steps of projected gradient (of the unbiased barrier problem) descent and dynamic programming,
                    and hopping across the fibers of the equilibrium bundle to a zero point of the canonical section.
              \item For any perfect equilibrium of any dynamic game, there are infinite many path on the equilibrium bundle that lead exactly to it.
                    The method achieves a weak approximation, namely, the approximation to an $\epsilon$-approximate perfect equilibrium, in fully polynomial time,
                    and the time complexity of achieving a strong approximation, namely, the approximation to an $\epsilon$-neighborhood of an actual perfect equilibrium, depends on the gradient of the canonical section near the actual perfect equilibrium.
              \item The computation of the algorithm is based on a variant of the expected utility instead of a specific game model, and thus the algorithm works with any game model that has a polynomial-time algorithm for the expected utility problem, and also with model-free cases.
              \item As for intermediate results, the game equilibrium problem is formalized as an optimization problem, which splits into two subproblems with respect to policy and value function.
                    \begin{itemize}
                        \item The subproblem with respect to policy is equivalent to the Nash equilibrium problem of static games. We introduce two equivalent concepts called unbiased barrier problem and unbiased KKT conditions,
                              making the interior point method to satisfy a so-called primal-dual unbiased condition using these two concepts,
                              so that it can approximate Nash equilibria by local optimization.
                        \item The subproblem with respect to value function is equivalent to the convergence problem of dynamic programming in dynamic games. We introduce a concept called policy cone, giving the iteration properties of dynamic programming and the equivalent conditions of equilibria in the context of policy cone,
                              and then use them to give the sufficient and necessary condition for dynamic programming to converge to perfect equilibria.
                    \end{itemize}
          \end{itemize}
\end{itemize}

Considering the research status and our major contributions, there are three direct impacts of our discovery.
\begin{itemize}
    \item Our discovery leads to a fundamental leap in the understanding of game equilibria through the equilibrium bundle.
    \item There is previously no efficient algorithms solving general game equilibrium problem, but now there is, and it is general enough to derive an efficient algorithm for any problem involving strategic interactions.
    \item PPAD=FP, where PPAD class is previously considered to contain hard problems, but our discovery shows that it actually does not, which means there is an efficient algorithm for every problem in the PPAD class to be discovered.
\end{itemize}

\section{Technical route and related work}
\subsection{Formalizing game equilibrium problem as optimization problem}
\begin{definition}[Regret minimization problem of dynamic games]
    A regret minimization problem of a dynamic game $\Gamma$ is optimization problem \eqref{regmin_problem},
    where $(\pi_a^{si},r_a^{si},V_s^i)$ is a tuple of policy, regret, and value,
    and $w^s$ may be any vector such that $w^s>0$.
    \begin{equation}
        \label{regmin_problem}
        \begin{aligned}
            \min_{(\pi_a^{si},r_a^{si},V_s^i,v_s^i)}\quad & w^s \pi_a^{si} r_a^{si}                                             \\
            \textrm{s.t.}\quad                            & r_a^{si}-v_s^i+\pi_{Aa}^{si-}(u_A^{si}+\gamma T_{s'A}^s V_{s'}^i)=0 \\
                                                          & \mathbf{1}_a\pi_a^{si}-\mathbf{1}^{si}=0                            \\
                                                          & v_s^i-V_s^i=0                                                       \\
                                                          & (\pi_a^{si},r_a^{si})\geq 0
        \end{aligned}
    \end{equation}
\end{definition}

The road to our FPTAS for perfect equilibria of dynamic games starts from the regret minimization problem \eqref{regmin_problem} of dynamic games.
The global optimal points of problem \eqref{regmin_problem} are equivalently the perfect equilibria and the optimal value is zero as shown in Theorem~\ref{per_equil_thm}.

\begin{proposition}
    \label{single_dp_lp}
    The following statements about linear programming problem \eqref{primal_lp} and \eqref{dual_lp} hold, where $\bar{w}^s$ may be any vector such that $\bar{w}^s>0$.
    \begin{equation}
        \label{primal_lp}
        \begin{aligned}
            \min_{V_s}\quad    & \bar{w}^s V_s                         \\
            \textrm{s.t.}\quad & V_s\geq u_a^s+\gamma T_{s'a}^s V_{s'} \\
        \end{aligned}
    \end{equation}
    \begin{equation}
        \label{dual_lp}
        \begin{aligned}
            \max_{\pi_a^s}\quad & \bar{w}^s V_{\pi s}              \\
            \textrm{s.t.}\quad  & \mathbf{1}_a\pi_a^s=\mathbf{1}^s \\
                                & \pi_a^s\geq 0
        \end{aligned}
    \end{equation}
    \begin{enumerate}
        \item Linear programming problem \eqref{primal_lp} and \eqref{dual_lp} are the dual problem of each other, and the optimal value is the optimal value function of the corresponding single-player dynamic game.
        \item The single-player case of regret minimization problem \eqref{regmin_problem} is the minimization of the duality gap of linear programming problem \eqref{primal_lp} and \eqref{dual_lp},
              where $\bar{w}^{s'}=w^s(I_{s'}^s-\gamma T_{\pi s'}^s)$.
    \end{enumerate}
\end{proposition}

MDP (Markov decision process) is the model for single-player decision problem\cite{mdp}, and an optimal policy of an MDP is the single-player degeneration of a perfect equilibrium of a dynamic game.
Value iteration $V_{s,k+1}=\max_a(u_a^s+\gamma T_{s'a}^s V_{s',k})$ and policy iteration $\pi_{a,k+1}^s=\arg\max_a(u_a^s+\gamma T_{s'a}^s V_{\pi_{k} s'})$ are two dynamic programming methods to compute the optimal policies of MDPs in existing studies.
Proposition~\ref{single_dp_lp}~(i) restates an existing result about the relation between dynamic programming and linear programming.
In addition, the two dynamic programming methods correspond to the two linear programmings, such that value iteration solves linear programming problem \eqref{primal_lp}, and policy iteration solves linear programming problem \eqref{dual_lp}, where the iteration variables correspond to the optimization variables.

Proposition~\ref{single_dp_lp}~(ii) shows that the single-player case of problem \eqref{regmin_problem} is in the form of the duality gap of a pair of dual linear programming problems, making alternately updating $\pi_a^{si}$ and $V_s^i$ feasible.
However, the general case of problem \eqref{regmin_problem} is a multilinear programming problem with joint constraints such that $\pi_a^{si}$ and $V_s^i$ simultaneously appear in a constraint, making it not possible to be decomposed with respect to $\pi_a^{si}$ and $V_s^i$.
Although decomposition with respect to $\pi_a^{si}$ and $V_s^i$ is unfeasible, we can still try to alternately update $\pi_a^{si}$ and $V_s^i$ by taking one of $\pi_a^{si}$ and $V_s^i$ as variable while the other as constant, splitting problem \eqref{regmin_problem} into two subproblems.
\begin{itemize}
    \item Regret minimization problem \eqref{regmin_equ} of static games with respect to $\pi_a^{si}$, which is equivalent to solving Nash equilibria of static games.
    \item Dynamic programming problem \eqref{dp_problem} with respect to $V_s^i$, which is equivalent to the convergence problem of dynamic programming in dynamic games.
\end{itemize}
Then, the perfect equilibrium problem is transformed to a set of Nash equilibrium problems indexed by states plus a dynamic programming problem.
In the next subsection, we introduce the technical routes of this paper to construct polynomial-time algorithms solving the two subproblems and combine them to a polynomial-time algorithm solving problem \eqref{regmin_problem}, which equivalently approximates perfect equilibria.

\subsection{Technical route}
In section \ref{ipm_sec}, we study the subproblem with respect to $\pi_a^{si}$ of problem \eqref{regmin_problem} and develope an algorithm solving it to update $\pi_a^{si}$.
The subproblem is regret minimization problem \eqref{regmin_equ}, and it is equivalent to the problem of computing Nash equilibria of static games.
\begin{itemize}
    \item {\bf Solving Nash equilibria is equivalent to solving the global optimal points of problem \eqref{regmin_equ}.}
          \begin{itemize}
              \item Existing methods like no-regret and self-play lack the ability to converge to Nash equilibria.
                    \begin{itemize}
                        \item According to the purification theorem\cite{pure_theo}, a policy is a Nash equilibrium if and only if for every player, every non-zero weighted action has the same utility.
                        \item As a consequence, mixed Nash equilibria are not stable equilibrium points for any method based on optimizing the utility of the updated mixed strategy, such as no-regret and self-play.
                        \item In no-regret\cite{noregret}, regret is a vector of differences between the utility of the updated mixed strategy and the utilities of every action, and the regret is minimized.
                        \item In self-play\cite{fp}, the utility of the updated mixed strategy is maximized.
                        \item In regret minimization problem \eqref{regmin_equ}, the regret is a vector of differences between the maximum utility and the utilities of every action, plus a scalar.
                    \end{itemize}
              \item {\bf Solving Nash equilibria is equivalent to solving the primal-dual unbiased local extreme points of problem \eqref{regmin_equ}.}
                    We introduce a primal-dual unbiased condition here, so that we can use local optimization methods that are polynomial-time.
              \item {\bf Using the interior point method to perform local optimization.}
                    The interior point method approximate a local extreme point along the central path through {\bf two levels of iteration}:
                    updating onto the central path by locally optimizating the barrier problem, and updating along the central path to a local extreme point by gradually reducing the barrier parameter to zero,
                    where the central path consists of solutions of the perturbed KKT conditions.
          \end{itemize}
    \item {\bf Using interior point method to approximate primal-dual unbiased local extreme points of problem \eqref{regmin_equ}.}
          We want the interior point method to always lead to a primal-dual unbiased local extreme point of problem \eqref{regmin_equ}, instead of just an ordinary local extreme point.
          \begin{itemize}
              \item {\bf Equivalently characterizing primal-dual unbiased local extreme points of the barrier problem of problem \eqref{regmin_equ} through unbiased barrier problem \eqref{ubarr_equ} or unbiased KKT conditions \eqref{ukkt_equ}.}
                    We introduce (primal-dual) unbiased barrier problem \eqref{ubarr_equ}, (primal-dual) unbiased KKT conditions \eqref{ukkt_equ}, and a Brouwer function here,
                    so that we can make the following transformation.
              \item {\bf Transforming the interior point method into a primal-dual unbiased interior point method.}
                    The primal-dual unbiased interior point method approximate a primal-dual unbiased local extreme point along the (primal-dual) unbiased central path through {\bf two levels of iteration},
                    where the unbiased central path consists of solutions of unbiased KKT conditions \eqref{ukkt_equ}.
                    \begin{itemize}
                        \item In the first iteration level, usually, the barrier problem is numerically optimized by SQP (sequential quadratic programming)\cite{interior_point}, which is a numerical second-order method.
                              But for unbiased barrier problem \eqref{ubarr_equ}, there exists an analytical expression of its projected gradient.
                              So we use projected gradient \eqref{projgrad_formula} to update onto the unbiased central path, and this analytical first-order method is much more simple.
                        \item In the second iteration level, we reduce the barrier parameter to update along the unbiased central path.
                    \end{itemize}
              \item {\bf There are three assumptions for the primal-dual unbiased interior point method to work.}
                    \begin{itemize}
                        \item Starting point: there has to be a starting point sufficiently close to the unbiased central path.
                        \item Differentiability: the updated point has to move an infinitesimal step along the unbiased central path with an infinitesimal reduction of the barrier parameter.
                        \item Convexity: unbiased barrier problem has to be strictly locally convex for projected gradient descent \eqref{projgrad_formula} to update onto the unbiased central path.
                    \end{itemize}
          \end{itemize}
    \item {\bf Introducing the equilibrium bundle to settle the three assumptions.}
          \begin{itemize}
              \item We introduce a geometric object called equilibrium bundle, along with its canonical section and singular point.
                    \begin{itemize}
                        \item Equilibrium bundle is a fiber bundle in differential geometry, which consists of disjoint fibers over every policy in the policy space.
                        \item {\bf Equilibrium bundle is the structured solution space of unbiased KKT conditions \eqref{ukkt_equ}.}
                        \item {\bf Nash equilibria are equivalently the zero points of the canonical section of the equilibrium bundle.}
                    \end{itemize}
              \item {\bf Settling the three assumptions on the equilibrium bundle.}
                    Equilibrium bundle corresponds to the previous unbiased central path in the following statements.
                    \begin{itemize}
                        \item Starting point: for any given policy, a large enough barrier parameter on its fiber is nearly on the equilibrium bundle.
                        \item Differentiability: on non-singular points of the equilibrium bundle, policy moves an infinitesimal step with an infinitesimal step of barrier parameter.
                        \item Convexity: on non-singular points of the equilibrium bundle, unbiased barrier problem \eqref{ubarr_equ} is strictly locally convex.
                        \item Singularity: for any given policy, a large enough barrier parameter on its fiber is guaranteed to be a non-singular point pairing with the policy.
                              Note that moving along a fiber changes neither the policy nor the canonical section, so this singular avoidance does not affect the line search to Nash equilibria.
                    \end{itemize}
              \item {\bf Transforming the primal-dual unbiased interior point method into a line search method on the equilibrium bundle.}
                    This method {\bf approximates any Nash equilibrium of any static game} through {\bf two levels of iteration}.
                    \begin{itemize}
                        \item In the first iteration level, we update onto the equilibrium bundle by projected gradient descent \eqref{projgrad_formula}.
                        \item In the second iteration level, singular points must be avoided as we update along the equilibrium bundle, which is done by canonical section descent \eqref{canosect_formula} leveraging the structure of the equilibrium bundle.
                              So we hop across the fibers of the equilibrium bundle to a zero point of its canonical section by canonical section descent \eqref{canosect_formula}.
                    \end{itemize}
          \end{itemize}
    \item To show the significance of our results, we additionally extend the existence and oddness theorems of Nash equilibria to the unbiased KKT conditions,
          such that the two theorems of Nash equilibria are implied by degenerating the two theorems of the unbiased KKT conditions.
          \begin{itemize}
              \item Existence theorem\cite{nash_orginal2}: for any barrier parameter, there is always a solution of the unbiased KKT conditions.
                    Because unbiased KKT conditions derive a Brouwer function mapping primal policies to dual policies for any barrier parameter, then a primal-dual unbiased policy always exists by Brouwer's fixed point theorem.
              \item Oddness theorem\cite{oddness2}: for any barrier parameter, there are almost always an odd number of solutions on the unbiased KKT conditions.
                    Because unbiased KKT conditions derive a polynomial function mapping policies to barrier parameters, then the equilibria and a unique starting point are almost always connected in pairs by Newton-Puiseux theorem.
          \end{itemize}
\end{itemize}

In section \ref{dp_sec}, we study the subproblem with respect to $V_s^i$ of problem \eqref{regmin_problem} and develope an algorithm solving it to update $V_s^i$.
The subproblem is dynamic programming problem \eqref{dp_problem}, and it is equivalent to the convergence problem of dynamic programming in dynamic games.
\begin{itemize}
    \item {\bf Solving dynamic programming problem \eqref{dp_problem} is equivalent to making the convergence point of dynamic programming to lie within the constraint of problem \eqref{dp_problem}.}
          \begin{itemize}
              \item Existing methods that directly use the Bellman operator in dynamic games lack the ability to converge at all.
                    \begin{itemize}
                        \item Value iteration is a dynamic programming method that uses the Bellman operator to iterate the value function,
                              and it is a polynomial-time exact algorithm for computing optimal policies of MDPs\cite{bellman}.
                              The iterative convergence of Bellman operator is contraction mapping convergence.
                        \item The Bellman operator is also directly used to solve perfect equilibria of dynamic games in early attempts\cite{mdp_game,nash_q}, showing that it generally do not converge.
                    \end{itemize}
              \item {\bf Settling for monotone convergence instead of contraction mapping convergence.}
                    \begin{itemize}
                        \item We introduce two dynamic programming operators respectively given by the objective and constraint of problem \eqref{dp_problem},
                              and introduce the policy cone and the best response cone respectively based on the two dynamic programming operators.
                        \item {\bf It is discovered that our dynamic programming operator achieves monotone convergence within the policy cone as shown in the following.}
                        \item Contraction mapping convergence can be considered as a convergence in every direction, and monotone convergence can be considered as a convergence in only a certain direction.
                        \item The Bellman operator is a special case of our dynamic programming operator, allowing us to analyze why Bellman operator does not converge in dynamic games later.
                    \end{itemize}
              \item {\bf For a fixed policy, dynamic programming operator achieves monotone convergence to the apex within the policy cone, and adding a large enough scalar keeps the iteration within the best response cone.}
                    \begin{itemize}
                        \item Policy cone is the monotonic and closed domain of dynamic programming operator such that monotone convergence theorem applies.
                        \item Every value function plus a large enough scalar is in the best response cone, and the best response cone lies inside the policy cone.
                        \item Policy cone has an apex that is the value function of its policy, and the apex is the limit point of the dynamic programming operator iteration, and plusing scalars during the iteration does not affect the convergence.
                    \end{itemize}
              \item {\bf A policy is a perfect equilibrium if and only if the apex of its policy cone lies in its best response cone, and if and only if the apex is the global optimal point of dynamic programming problem \eqref{dp_problem}.}
          \end{itemize}
    \item {\bf Making dynamic programming operator to monotonically converge to a policy cone apex that lies in the best response cone.}
          \begin{itemize}
              \item {\bf Explaining why Bellman operator cannot be simply generalized from MDPs to dynamic games.}
                    Bellman operator achieves both contraction mapping convergence and monotone convergence in single-player case, but it achieves neither the convergence in multi-player case.
              \item {\bf The sufficient and necessary condition for dynamic programming iteration \eqref{dp_formula} to converges to a perfect equilibrium value function.}
                    \begin{itemize}
                        \item Not only is the policy cone bounded below given a fixed policy, but the apex of it is also bounded below with respect to a varying policy.
                              Thus, monotone convergence with a varying policy is maintained by adding a large enough scalar so that the iteration lies within the policy cone and letting this scalar converge to zero.
                        \item The limit point lies in the best response cone if the iteration is kept within the best response cone.
                              Thus, we further require the scalar large enough so that the iteration lies within the best response cone.
                    \end{itemize}
          \end{itemize}
\end{itemize}

In section \ref{hybrid_sec}, we combine the above two algorithms by alternating the steps updating $\pi_a^{si}$ and $V_s^i$ such that the resulting algorithm solves problem \eqref{regmin_problem},
which is equivalent to the problem of computing perfect equilibria of dynamic games.
\begin{itemize}
    \item {\bf Perfect equilibria are equivalently the global optimal points of regret minimization problem \eqref{regmin_problem} of dynamic games.}
          \begin{itemize}
              \item Existing methods that directly use variants of the static game method in dynamic games lack the ability to converge to perfect equilibria.
                    \begin{itemize}
                        \item For every dynamic game, there is a corresponding static game, in which the behavioral strategies of the dynamic game serve as the static strategies of the corresponding static game.
                        \item When the variants of no-regret\cite{cfr} and self-play\cite{fsp} are used in dynamic games, they deal with interactions of behavioral strategies, instead of static strategies at each stage of the dynamic game.
                              This means that their resulting equilibria are at most Nash equilibria of the corresponding static game of behavioral strategy interactions, while not guaranteed to be a Nash equilibrium at every stage of the dynamic game.
                    \end{itemize}
              \item {\bf Combining the line search on equilibrium bundle and the dynamic programming within policy cone.}
                    \begin{itemize}
                        \item {\bf A policy is a perfect equilibrium if and only if it is the global optimal point of problem \eqref{regmin_problem}.}
                              There is a relation between canonical section and the two dynamic programming operators, which then establishs a relation between the results in the last two sections.
                        \item {\bf Determining the nesting order of the three iteration formulas of the two sub-algorithms.}
                              Dynamic programming \eqref{dp_formula} has to be at the same level of iteration as the projected gradient descent \eqref{projgrad_formula},
                              because dynamic programming moves more than an infinitesimal step every iteration, which would cause the updated point to be too far away from the equilibrium bundle to update back,
                              except for when dynamic programming is nearly converged.
                        \item {\bf Defining the unbiased KKT conditions and the equilibrium bundle of dynamic games.}
                              By setting the value function as the convergence point of dynamic programming, which is the value function of a specific policy,
                              there are still unbiased KKT conditions and the equilibrium bundle of dynamic games, such that those of static games are the single-state degenerations.
                              Furthermore, {\bf the existence and oddness theorems generalize to dynamic games.}
                        \item The combined method is still formalized as {\bf a line search on the equilibrium bundle},
                              which {\bf approximates any perfect equilibrium of any dynamic game} through {\bf two levels of iteration}:
                              updating onto the equilibrium bundle by projected gradient descent \eqref{projgrad_formula} and dynamic programming \eqref{dp_formula},
                              and hopping across the fibers of the equilibrium bundle to a zero point of its canonical section by canonical section descent \eqref{canosect_formula}.
                    \end{itemize}
              \item Finally, we show that the resulting algorithm achieves a weak approximation in fully polynomial time by showing the convergence rates of the three iterations \eqref{projgrad_formula}, \eqref{dp_formula}, and \eqref{canosect_formula}.
          \end{itemize}
\end{itemize}

\section{Interior point method on static games}\label{ipm_sec}
\subsection{Approximating Nash equilibrium by local optimization}\label{ipm_sec1}
In this section, we deal with the subproblem with respect to $\pi_a^{si}$ of problem \eqref{regmin_problem}, which is the regret minimization problem \eqref{regmin_equ} with $\mu_a^i=0$ and $U_A^i=(u_A^{si}+\gamma T_{s'A}^s V_{s'}^i)(x)$.
The global optimal points of problem \eqref{regmin_equ} are equivalently Nash equilibria, and thus the goal of this section is to find a polynomial-time algorithm for Nash equilibria.
We first use the purification theorem to explain why existing approximation methods for Nash equilibrium lack convergence guarantee while methods based on regret minimization problem \eqref{regmin_equ} do.
Then we give Theorem~\ref{regmin_equiv} stating that a Nash equilibrium is equivalently a global optimal point of problem \eqref{regmin_equ},
and is also equivalently a local extreme point that also satisfies a condition called primal-dual unbiased condition.
Finally, the second equivalent condition would allow us to approximate Nash equilibria by local optimization, which is chosen as the interior point method.
In addition, this section mainly deals with static games, so we drop the state index $s$ as we state before.

\begin{definition}[Regret minimization problem]
    Let $G$ be a static game, where the utility function is $U:\prod_{i\in N} \mathcal{A} \to \prod_{i\in N} \mathbb{R}$.
    A regret minimization problem of $G$ is optimization problem \eqref{regmin_equ} with $\mu_a^i=0$,
    where $(\pi_a^i,r_a^i,v^i)$ is a tuple of policy, regret, and value, and $\mu_a^i$ is a barrier parameter.
    \begin{equation}
        \label{regmin_equ}
        \begin{aligned}
            \min_{(\pi_a^i,r_a^i,v^i)}\quad & \pi_a^i r_a^i-\mu_a^i\ln r_a^i-\mu_a^i\ln \pi_a^i \\
            \textrm{s.t.}\quad              & r_a^i-v^i+\pi_{Aa}^{i-}U_A^i=0                    \\
                                            & \mathbf{1}_a\pi_a^i-\mathbf{1}^i=0
        \end{aligned}
    \end{equation}
    The case where $\mu_a^i> 0$ is a barrier problem of regret minimization problem \eqref{regmin_equ}.
\end{definition}

According to the purification theorem\cite{pure_theo}, $\pi_a^i$ is a Nash equilibrium if and only if for every player $i$, every action $a$ that satisfies $\pi_a^i>0$ has the same utility.
Thus, mixed Nash equilibria are not stable equilibrium points for any method based on optimizing the utility of the updated mixed strategy, such as no-regret\cite{noregret,cfr} and self-play\cite{fp,fsp}.
Note that the meanings of regret are different in no-regret methods and in this paper.
The regret in no-regret methods is a vector of differences between the utility of the updated mixed strategy and the utilities of every action,
while the regret in regret minimization problem \eqref{regmin_equ} is a vector of differences between the maximum utility and the utilities of every action plus a scalar.
The difference upon the quantity being optimized is why our method has convergence guarantee while existing methods based on no-regret or self-play do not,
and consequently exhibit non-stationarity when used in MARL.

\begin{theorem}
    \label{regmin_equiv}
    Let $G$ be a static game and $(\pi_a^i,r_a^i,v^i)$ be a tuple of policy, regret and value of $G$.
    Then the following statements are equivalent.
    \begin{enumerate}
        \item $(\pi_a^i,v^i)$ is a Nash equilibrium of $G$.
        \item $(\pi_a^i,r_a^i,v^i)$ is a global optimal point of regret minimization problem \eqref{regmin_equ}.
        \item There exist Lagrangian multipliers $(\bar{\lambda}_a^i,\tilde{\lambda}^i,\hat{\pi}_a^i,\hat{r}_a^i)$ that satisfies two conditions:
              the KKT conditions shown by equation \eqref{perturbed_kkt} with $\mu_a^i=0$, and the primal-dual unbiased condition $\pi_a^i=\hat{\pi}_a^i$.
              \begin{equation}
                  \label{perturbed_kkt}
                  \begin{bmatrix}
                      \bar{\lambda}_a^j\pi_{Aaa'}^{ij-}U_A^i+\tilde{\lambda}^i\mathbf{1}_{a'}+r_{a'}^i-\hat{r}_{a'}^i \\
                      \bar{\lambda}_a^i+\pi_a^i-\hat{\pi}_a^i                                                         \\
                      -\mathbf{1}_a\bar{\lambda}_a^i                                                                  \\
                      r_a^i\circ \hat{\pi}_a^i-\mu_a^i                                                                \\
                      \pi_a^i\circ \hat{r}_a^i-\mu_a^i                                                                \\
                      r_a^i-v^i+\pi_{Aa}^{i-}U_A^i                                                                    \\
                      \mathbf{1}_a\pi_a^i-\mathbf{1}^i
                  \end{bmatrix}=0
              \end{equation}
              The case where $\mu_a^i> 0$ of equation \eqref{perturbed_kkt} is the perturbed KKT conditions of regret minimization problem \eqref{regmin_equ}.
    \end{enumerate}
    Furthermore, when these statements hold, the objective function $\pi_a^i r_a^i$ of \eqref{regmin_equ} is $0$,
    and $(\bar{\lambda}_a^i,\tilde{\lambda}^i,\pi_a^i-\hat{\pi}_a^i,r_a^i-\hat{r}_a^i)=0$.
\end{theorem}

Theorem~\ref{regmin_equiv} points out that the Nash equilibrium is not only equivalent to the global optimal point of regret minimization problem \eqref{regmin_equ},
but also equivalent to its local extreme point that satisfies the primal-dual unbiased condition, where the local extreme point is namely the point that satisfies the KKT conditions\cite{kkt_orginal} in (iii).
In previous research\cite{regmin1,regmin2}, linear programming, quadratic programming, linear complementarity problem, and Nash equilibrium of bimatrix games have been unified,
and the unification is actually a degeneration of $(i)\Leftrightarrow(ii)$ into two-player static games.
$(i)\Leftrightarrow(ii)$ points out the connection between our method and existing methods, but it is not necessary in constructing our method.
$(i)\Leftrightarrow(iii)$ makes it possible to use local optimization that leads to local extreme points where primal-dual bias $\pi_a^i-\hat{\pi}_a^i$ is $0$ to compute global optimal points of \eqref{regmin_equ}, and equivalently Nash equilibria.

We choose \textbf{interior point method} to perform the local optimization, where \textbf{central path} is a path in the solution space of perturbed KKT conditions \eqref{perturbed_kkt} with the barrier parameter $\mu_a^i=\mathbf{1}_a^i\mu$,
the \textbf{first iteration level} is to update onto the central path by locally optimizating the barrier problem \eqref{regmin_equ},
the \textbf{second iteration level} is to update along the central path by reducing the scalar barrier parameter $\mu$,
and the central path leads to a local extreme point of the orginal problem as $\mu$ reduces to $0$.
In our case, we additionally intend to keep the update on a particular central path on which primal-dual bias $\pi_a^i-\hat{\pi}_a^i$ is $0$.

\subsection{Unbiased barrier problem and unbiased KKT conditions}\label{ipm_sec2}
Note that Theorem~\ref{regmin_equiv} (iii) only applies to the $\mu_a^i=0$ case, but we have to study the $\mu_a^i> 0$ case to find a way to keep the update on the central path where primal-dual bias is $0$.
Thus, we introduce two concepts called unbiased barrier problem and unbiased KKT conditions,
and give Theorem~\ref{ipm_equiv} to extend our equivalent result about the $\mu_a^i=0$ case in Theorem~\ref{regmin_equiv} (iii) onto the subject of interior point method where $\mu_a^i> 0$,
so that we can construct a primal-dual unbiased interior point method.
In addition, we use these two concepts to give an extension of the existence theorem of Nash equilibria.

\begin{figure}
    \centering
    \includegraphics[width=0.9\textwidth]{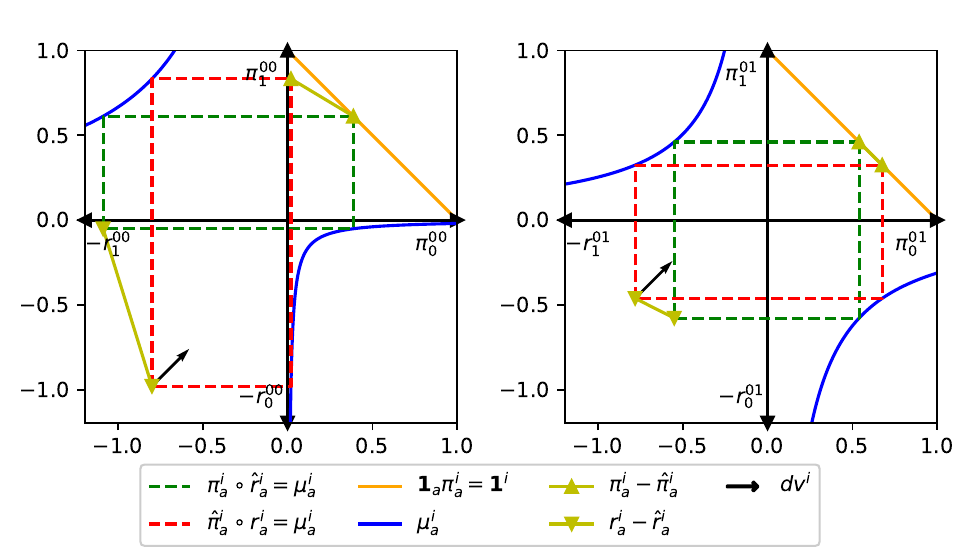}
    \caption{
        Graph of unbiased barrier problem.
        This figure is plotted with a dynamic game where $N=\mathcal{A}=\{0,1\},\mathcal{S}=\{0\}$.
        The graph is based on the joint space of policy $\pi_a^i$ and regret $r_a^i$.
        The positive half of the two axes represent two action indices $a\in\{0,1\}$ of $\pi_a^i$,
        the negative half of the two axes represent two action indices $a\in\{0,1\}$ of $r_a^i$,
        and the two subfigures represent two player indices $i\in\{0,1\}$.
        Plotting $\pi_a^i$ and $\hat{\pi}_a^i$ on the all positive orthant, $r_a^i$ and $\hat{r}_a^i$ on the all negative orthant, and $\mu_a^i$ between positive half and negative half of the axes as hyperbolas,
        $\hat{\pi}_a^i\circ r_a^i=\mu_a^i$ and $\pi_a^i\circ\hat{r}_a^i=\mu_a^i$ have rectangular shapes, and $(\pi_a^i-\hat{\pi}_a^i,r_a^i-\hat{r}_a^i)$ is the bias of two rectangles.
        $dv^i$ is the direction $r_a^i$ can move with fixed $\pi_a^i$ and within the constraint $r_a^i=v^i-\pi_{Aa}^{i-}U_A^i$ of unbiased barrier problem \eqref{ubarr_equ}.
        With $r_a^i$ moving in direction $dv^i$ and the other two corners of the rectangle fixed on the hyperbolas, the figure illustrates how there is a unique $v^i$ to let $\hat{\pi}_a^i$ satisfy $\mathbf{1}_a\hat{\pi}_a^i=\mathbf{1}^i$ as stated in Theorem~\ref{ipm_theo}~(i), and the right subfigure shows a case where it is satisfied.
    }
    \label{barrproblem_graph}
\end{figure}
\begin{figure}
    \centering
    \includegraphics[width=0.9\textwidth]{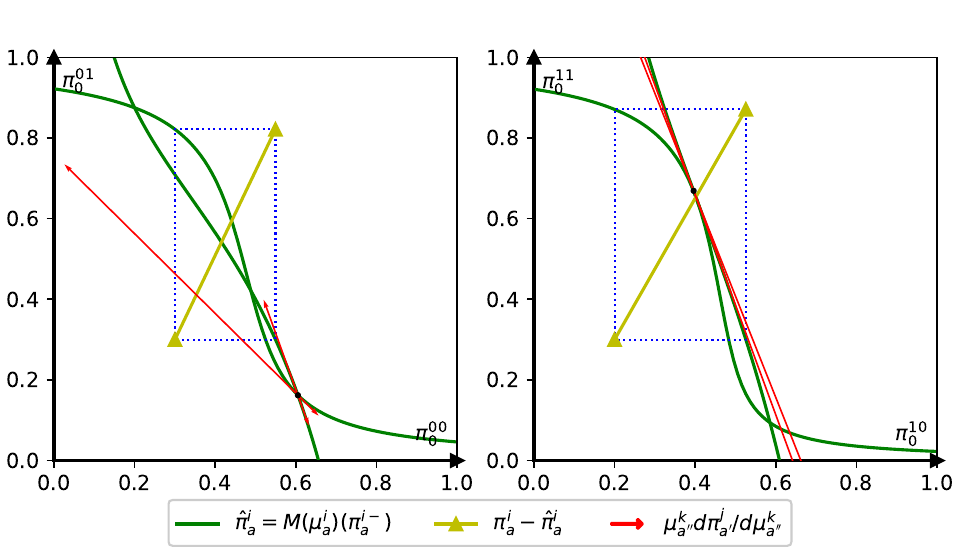}
    \caption{
        Graph of unbiased KKT conditions.
        This figure is plotted with a dynamic game where $N=\mathcal{S}=\mathcal{A}=\{0,1\}$.
        The graph is based on the policy space, where the two axes represent two player indices $i\in\{0,1\}$ of $\pi_a^{si}$, the two subfigures represent two state indices $s\in\{0,1\}$ of $\pi_a^{si}$,
        and only one of the two action indices $a=0$ is needed to represent $\pi_a^{si}$ since $\pi_a^{si}$ sums to $1$ over action indices.
        $\hat{\pi}_a^i=M(\mu_a^i)(\pi_a^{i-})$ is a set of hypersurfaces indexed by $i\in N$ and induced by the Brouwer function $M(\mu_a^i)$ for a given $\mu_a^i$, $\pi_a^i-\hat{\pi}_a^i$ shows the mapping of $\hat{\pi}_a^i=M(\mu_a^i)(\pi_a^i)$, and the intersections of the hypersurfaces are fixed points of $M(\mu_a^i)$.
        There is at least one intersection of the hypersurfaces according to Theorem~\ref{ipm_theo}~(iii), and there are almost always an odd number of intersections according to Theorem~\ref{odd_thm} (ii), as extensions of the existence and oddness theorems of Nash equilibria.
        Differential $(d\pi_{a'}^j/\pi_{a'}^j)/(d\mu_{a''}^k/\mu_{a''}^k)$ illustrates that a intersection $\pi_a^i$ moves with the $i$-th hypersurface as $\mu_a^i$ varies on index $i$, since the $i$-th hypersurface is only relevant to $\mu_a^i$ on index $i$.
        The right subfigure shows a singular point in Definition~\ref{equil_bund_def}, where the differential grows infinite large.
    }
    \label{kktcondition_graph}
\end{figure}

\begin{definition}[Unbiased barrier problem]
    \label{ubarr_def}
    An unbiased barrier problem is the optimization problem
    \begin{equation}
        \label{ubarr_equ}
        \begin{aligned}
            \min_{(\pi_a^i,r_a^i,v^i)}\quad & \left(\pi_a^i-\hat{\pi}_a^i\right)\left(r_a^i-\hat{r}_a^i\right) \\
            \textrm{s.t.}\quad              & r_a^i-v^i+\pi_{Aa}^{i-}U_A^i=0                                   \\
                                            & \mathbf{1}_a\pi_a^i-\mathbf{1}^i=0
        \end{aligned}
    \end{equation}
    parameterized by $\hat{\pi}_a^i$ and $\hat{r}_a^i$, and $\hat{\pi}_a^i=\mu_a^i/r_a^i$ and $\hat{r}_a^i=\mu_a^i/\pi_a^i$, which are called dual policy and dual regret respectively,
    and the tuple $(\pi_a^i-\hat{\pi}_a^i,r_a^i-\hat{r}_a^i)$ is called primal-dual bias.
\end{definition}
\begin{definition}[Unbiased KKT conditions]
    \label{ukkt_def}
    Unbiased KKT conditions are simultaneous equations
    \begin{subequations}
        \label{ukkt_equ}
        \begin{equation}
            \label{ukkt_equ1}
            \begin{bmatrix}
                \hat{\pi}_a^i\circ r_a^i-\mu_a^i \\
                r_a^i-v^i+\pi_{Aa}^{i-}U_A^i     \\
                \mathbf{1}_a\hat{\pi}_a^i-\mathbf{1}^i
            \end{bmatrix}=0,
        \end{equation}
        \begin{equation}
            \pi_a^i=\hat{\pi}_a^i.
        \end{equation}
    \end{subequations}
\end{definition}
\begin{definition}[Brouwer function]
    \label{brouwer_def}
    Brouwer function is the family of maps $M:\{\mu_a^i|\mu_a^i\geq 0\}\to(\mathcal{P}\to \mathcal{P})$ parameterized by $\mu_a^i$ such that $M(\mu_a^i)(\pi_a^i)=\hat{\pi}_a^i$ satisfies equation \eqref{ukkt_equ1}.
\end{definition}

\begin{theorem}
    \label{ipm_equiv}
    Given $\mu_a^i>0$, for the tuple $(\pi_a^i,r_a^i,v^i)$, the following properties satisfy $(i)\Leftrightarrow(ii)\Leftrightarrow(iii)\Leftrightarrow(\pi_a^i=\hat{\pi}_a^i \wedge (iv))$,
    and also satisfy $(iv)\Leftrightarrow(v)\Leftrightarrow(vi)$.
    \begin{enumerate}
        \item Being a fixed point of Brouwer function $\hat{\pi}_a^i=M(\mu_a^i)(\pi_a^i)$.
        \item Being a global optimal point of unbiased barrier problem \eqref{ubarr_equ}.
        \item Being a solution of unbiased KKT conditions \eqref{ukkt_equ}.
        \item Being a solution of perturbed KKT conditions \eqref{perturbed_kkt} for some $(\bar{\lambda}_a^i,\tilde{\lambda}^i,\hat{\pi}_a^i,\hat{r}_a^i)$.
        \item Being a KKT point of unbiased barrier problem \eqref{ubarr_equ}.
        \item Being a local extreme point of barrier problem \eqref{regmin_equ}.
    \end{enumerate}
    In particular, $\pi_a^i$ is a Nash equilibrium if and only if $(\pi_a^i,\mathbf{0}_a^i)$ is a solution of unbiased KKT conditions \eqref{ukkt_equ}.
\end{theorem}

The Brouwer function connects unbiased barrier problem and unbiased KKT conditions.
The unbiased KKT conditions derive the Brouwer function mapping primal policies to dual policies, such that the fixed points of the Brouwer function are solutions of the unbiased KKT conditions.
The unbiased barrier problem depicts the approximation of fixed points of the Brouwer function, such that the fixed points are the global optimal points, or zero points, of the unbiased barrier problem.

There are two equivalence formulas in Theorem~\ref{ipm_equiv}.
The first one $(i)\Leftrightarrow(ii)\Leftrightarrow(iii)\Leftrightarrow((iv)\wedge \pi_a^i=\hat{\pi}_a^i)$ establishs the equivalence between the unbiased barrier problem, Brouwer function, unbiased KKT conditions, and the $\mu_a^i> 0$ case in Theorem~\ref{regmin_equiv} (iii), allowing us the construct a primal-dual unbiased interior point method.
Actually, we only need the first formula to make subsequent constructions,
while the second formula $(iii)\Leftrightarrow(iv)\Leftrightarrow(v)$ only illustrates the connection between our method and interior point method, that is, a path that consists of solutions of the unbiased KKT conditions really is a central path in interior point methods on which the primal-dual bias is $0$, ignoring the difference that the barrier parameter in our method is actually a vector instead of a scalar.

\begin{theorem}
    \label{ipm_theo}
    The following properties about unbiased barrier problem \eqref{ubarr_equ} and unbiased KKT conditions \eqref{ukkt_equ} hold.
    \begin{enumerate}
        \item For any $\pi_a^i$ and $\mu_a^i$, there exists a unique $v^i$ that satisfies equation \eqref{ukkt_equ1}.
        \item If $\mathbf{1}_a\hat{\pi}_a^i-\mathbf{1}^i=0$, then the projected gradient of \eqref{ubarr_equ} satisfies
              \begin{equation}
                  \label{projgrad_formula}
                  d\left(\left(\pi_a^i-\hat{\pi}_a^i\right)\left(r_a^i-\hat{r}_a^i\right)\right)=\left(\pi_a^i-\hat{\pi}_a^i\right)\left({\rm Diag}(r_a^i)-\pi_{Aaa'}^{ij-}U_A^i\circ\pi_{a'}^j\right)\left(d\pi_{a'}^j/\pi_{a'}^j\right).
              \end{equation}
        \item {\rm (Existence theorem)} For every $\mu_a^i$, there is at least one solution of \eqref{ukkt_equ}.
    \end{enumerate}
\end{theorem}

Theorem~\ref{ipm_theo}~(i) shows two facts.
First, (i) shows that there is a unique $v^i$ to let dual policy $\hat{\pi}_a^i$ satisfy constraint $\mathbf{1}_a\pi_a^i=\mathbf{1}^i$.
As a consequence, (ii) shows that when $\mathbf{1}_a\hat{\pi}_a^i=\mathbf{1}^i$, $\pi_a^i-\hat{\pi}_a^i$ and $dv^i$ are orthogonal,
and the projected gradient of unbiased barrier problem \eqref{ubarr_equ} has the form of an analytical expression as equation \eqref{projgrad_formula} shows.
Second, (i) shows that for every policy $\pi_a^i$, there is a unique dual policy $\hat{\pi}_a^i$ that satisfies \eqref{ukkt_equ1}.
As a consequence, the Brouwer function $\hat{\pi}_a^i=M(\mu_a^i)(\pi_a^i)$ is indeed a map,
then by Brouwer's fixed point theorem, (iii) asserts the existence of a solution of unbiased KKT conditions \eqref{ukkt_equ} for every $\mu_a^i$, as an extension of the existence of Nash equilibria\cite{nash_orginal2}.

Note that equation \eqref{projgrad_formula} specifies the gradient with respect to $\ln\pi_a^i$, and it needs to be projected orthogonally to the hyperplane of $\mathbf{1}_a\pi_a^i=\mathbf{1}^i$.
However, we use $\pi_a^i=\exp(\sigma_a^i)/(\mathbf{1}_a\exp(\sigma_a^i))$ as the representation of policy in our experiments,
and the gradient with respect to $\sigma_a^i$ is ${\rm pg}_{a''}^j$ in the following equation, which naturally satisfies constraint $\mathbf{1}_a\pi_a^i=\mathbf{1}^i$ without projection.
In general, the gradient may be different for different representation of $\pi_a^i$, and it generally needs to be projected.
\begin{equation}
    {\rm pg}_{a''}^j:=\left(\pi_a^i-\hat{\pi}_a^i\right)\left({\rm Diag}(r_a^i)-\pi_{Aaa'}^{ij-}U_A^i\circ\pi_{a'}^j\right)\left(I_{a'a''}-\mathbf{1}_{a'}\pi_{a''}^j\right)
\end{equation}

In summary, we've transformed the interior point method into a \textbf{primal-dual unbiased interior point method},
where \textbf{(primal-dual) unbiased central path} is a path in the solution space of unbiased KKT conditions \eqref{ukkt_equ},
the \textbf{first iteration level} is to update onto the unbiased central path by projected gradient descent \eqref{projgrad_formula},
the \textbf{second iteration level} is to update along the unbiased central path by reducing barrier parameter $\mu_a^i$,
and the unbiased central path leads to a Nash equilibrium as $\mu_a^i$ reduces to $0$.

However, \textbf{there are three assumptions} for the primal-dual unbiased interior point method to work.
\textbf{Starting point}: the iteration has to start from a point that is sufficiently close to a known point on the unbiased central path.
\textbf{Differentiability}: in the second iteration level, the policy $\pi_a^i$ has to move an infinitesimal step along the unbiased central path with an infinitesimal reduction of the barrier parameter $\mu_a^i$, so that the first iteration level can update back onto the unbiased central path.
\textbf{Convexity}: in the first iteration level, unbiased barrier problem \eqref{ubarr_equ} has to be strictly locally convex for projected gradient descent \eqref{projgrad_formula} to update onto the unbiased central path.

\subsection{Equilibrium bundle}\label{ipm_sec3}
Settling the three assumptions is beyond the capacity of mathematical optimization, it turns out to be a geometric problem regarding the solution space of unbiased KKT conditions \eqref{ukkt_equ} where unbiased central paths lie.
Thus, we first introduce a geometric object called equilibrium bundle to structure the solution space of unbiased KKT conditions \eqref{ukkt_equ},
so that we can study the equilibrium bundle instead of the unbiased central path.
Then we settle the three assumptions using the geometric properties of the equilibrium bundle, transforming our primal-dual unbiased interior point method into a line search method on the equilibrium bundle.
Finally, we give the oddness theorem of the equilibrium bundle as an extension of that of Nash equilibria.

\begin{definition}[Equilibrium bundle]
    \label{equil_bund_def}
    An equilibrium bundle of a static game $G$ is the tuple $(E , \mathcal{P}, \alpha:E\to \mathcal{P})$ given by the following equations, where $\mathcal{P}=\prod_{i\in N}\Delta(\mathcal{A})$ is the policy space.
    \begin{equation}
        \label{equil_bund_equ}
        \begin{aligned}
             & E=\bigcup _{\pi_a^i \in \mathcal{P}} \{\pi_a^i\} \times B(\pi_a^i)                             \\
             & B(\pi_a^i)=\left\{v^i\circ\pi_a^i+\bar{\mu}_a^i(\pi_a^i)|v^i\geq 0\right\}                     \\
             & \bar{\mu}_a^i(\pi_a^i)=\pi_a^i \circ \left(\max_a \pi_{Aa}^{i-}U_A^i-\pi_{Aa}^{i-}U_A^i\right) \\
             & \alpha\left((\pi_a^i,\mu_a^i)\right) =\pi_a^i                                                  \\
        \end{aligned}
    \end{equation}
    The map $\bar{\mu}_a^i:\mathcal{P}\to \{\mu_a^i|\min_a^i\mu_a^i=\mathbf{0}^i\}$ is called the canonical section of the equilibrium bundle.
    The point where ${\rm C}_{(j,a')\cup l}^{(i,a)\cup m}$ is singular is called a singular point of it, where ${\rm C}_{(j,a')\cup l}^{(i,a)\cup m}$ is given by equation \eqref{coeff_matrix}, .
\end{definition}

\begin{theorem}
    \label{equil_bund_equiv}
    Let $(E , \mathcal{P}, \alpha:E\to \mathcal{P})$ be the equilibrium bundle of static game $G$, then the following statements hold.
    \begin{enumerate}
        \item $(\pi_a^i,\mu_a^i)\in E$ if and only if $(\pi_a^i,\mu_a^i)$ is a solution of unbiased KKT conditions \eqref{ukkt_equ}.
        \item Given $\pi_a^i$, then $\mu_a^i\in B(\pi_a^i)$ if and only if $(\pi_a^i,\mu_a^i)$ is a solution of unbiased KKT conditions \eqref{ukkt_equ}.
        \item $\pi_a^i$ is a Nash equilibrium if and only if the canonical section $\bar{\mu}_a^i(\pi_a^i)=0$.
    \end{enumerate}
\end{theorem}

Definition~\ref{equil_bund_def} and Theorem~\ref{equil_bund_equiv} show the relation between the equilibrium bundle, the solution space of unbiased KKT conditions \eqref{ukkt_equ}, and the joint space $\mathcal{P}\times\{\mu_a^i|\mu_a^i\geq 0\}$ of policy and barrier parameter.
First, the equilibrium bundle is the solution space of unbiased KKT conditions \eqref{ukkt_equ}, and a fiber $B(\pi_a^i)$ is the solution subspace relating to a given $\pi_a^i$.
Second, $B(\pi_a^i)$ is the intersection between $\{\mu_a^i|\mu_a^i\geq 0\}$ and a $\left\lvert N\right\rvert $-dimensional affine space spanned by $v^i$, and the union of all the $B(\pi_a^i)$ is exactly $\{\mu_a^i|\mu_a^i\geq 0\}$ by the existence theorem.
Third, the canonical section $\bar{\mu}_a^i(\pi_a^i)$ is the least element in $B(\pi_a^i)$ such that for any $\mu_a^i\in B(\pi_a^i)$ there is $\mu_a^i\geq \bar{\mu}_a^i(\pi_a^i)$, and $\bar{\mu}_a^i(\pi_a^i)$ is on the boundary of $\{\mu_a^i|\mu_a^i\geq 0\}$.
Finally, Theorem~\ref{equil_bund_equiv} (iii) shows that Nash equilibria are exactly the zero points of the canonical section $\bar{\mu}_a^i$.

Theorem~\ref{equil_bund_equiv} (iii) has two implications.
First, unlike the primal-dual unbiased interior point method, $\mu_a^i$ does not have to decrease to $0$ to reach a Nash equilibrium,
instead, $\mu_a^i$ can be any value on the fiber as long as the canonical section $\bar{\mu}_a^i(\pi_a^i)$ decreases to $0$.
Second, the canonical section $\bar{\mu}_a^i(\pi_a^i)$ depicts the global distribution of Nash equilibria, which can be used to search the policy space globally for the entire set of Nash equilibria.

In summary, the equilibrium bundle is exactly the structured solution space of unbiased KKT conditions \eqref{ukkt_equ} where the unbiased central paths lie,
making it possible to leverage the geometric properties of the equilibrium bundle to study the three assumptions left to settle in the primal-dual unbiased interior point method, which gives the following theorem.
\begin{equation}
    \label{coeff_matrix}
    {\rm C}_{(j,a')\cup l}^{(i,a)\cup m}:=
    \begin{bmatrix}
        H_{(j,a')}^{(i,a)}   & \hat{B} ^{(i,a)}_l \\
        \check{B}_{(j,a')}^m & \mathbf{0}_l^m
    \end{bmatrix}:=
    \begin{bmatrix}
        {\rm Diag}\left(r_a^i\right)_{(j,a')}^{(i,a)}-\left(\pi_{Aaa'}^{ij-}U_A^i\circ\pi_{a'}^j\right)_{(j,a')}^{(i,a)} & \left(I^{il}\mathbf{1}_a\right) _l^{(i,a)} \\
        \pi_{(j,a')}\circ \left(I^{jm}\mathbf{1}_{a'}\right) _{(j,a')}^m                                                 & \mathbf{0}_l^m
    \end{bmatrix}
\end{equation}

\begin{theorem}
    \label{equil_bund_theo}
    The following properties about equilibrium bundle $E$ hold.
    \begin{enumerate}
        \item For any $\hat{\mu}_a^i>0$, the algebraic curve ${\rm AC}=\{(\pi_a^i,\hat{\mu}_a^i\mu')\in E|\mu'> 0\}$ satisfies $$\lim_{\mu' \to +\infty} \pi_a^i=\frac{\hat{\mu}_a^i}{\mathbf{1}_a\hat{\mu}_a^i}.$$
        \item The differential $(d\pi_{a'}^j/\pi_{a'}^j)/(d\mu_{a''}^k/\mu_{a''}^k)$ at $(\pi_a^i,\mu_a^i)\in E$ of equilibrium bundle $E$ satisfies
              \begin{equation}
                  \label{tan_vec}
                  \begin{bmatrix}
                      H_{(j,a')}^{(i,a)}   & \hat{B} ^{(i,a)}_l \\
                      \check{B}_{(j,a')}^m & \mathbf{0}_l^m
                  \end{bmatrix}
                  \begin{bmatrix}
                      \left(\frac{\mu_{a''}^k d\pi_{a'}^j}{\pi_{a'}^j d\mu_{a''}^k}\right)_{(k,a'')}^{(j,a')} \\\left(\frac{\mu_{a''}^k dv^l}{d\mu_{a''}^k}\right)_{(k,a'')}^l
                  \end{bmatrix}=
                  \begin{bmatrix}
                      {\rm Diag}\left(r_a^i\right)_{(k,a'')}^{(i,a)} \\\mathbf{0}_{(k,a'')}^m
                  \end{bmatrix}.
              \end{equation}
              It follows that the equilibrium bundle is differentiable at $(\pi_a^i,\mu_a^i)\in E$ if and only if ${\rm C}_{(j,a')\cup l}^{(i,a)\cup m}$ is non-singular.
        \item The differential of unbiased barrier problem \eqref{ubarr_equ} satisfies
              \begin{equation}
                  \label{grad}
                  \begin{bmatrix}
                      \left(\pi_a^i-\hat{\pi}_a^i\right)^{(i,a)} & \mathbf{0}^m
                  \end{bmatrix}
                  \begin{bmatrix}
                      H_{(j,a')}^{(i,a)}   & \hat{B} ^{(i,a)}_l \\
                      \check{B}_{(j,a')}^m & \mathbf{0}_l^m
                  \end{bmatrix}
                  \begin{bmatrix}
                      \left(d\pi_{a'}^j/\pi_{a'}^j\right)^{(j,a')} \\dv^l
                  \end{bmatrix}.
              \end{equation}
              It follows that unbiased barrier problem \eqref{ubarr_equ} is locally strictly convex at $(\pi_a^i,\mu_a^i)\in E$ if and only if ${\rm C}_{(j,a')\cup l}^{(i,a)\cup m}$ is non-singular.
        \item For any $\pi_a^i$, there exists $\check{\mu}_a^i$ on its fiber $B(\pi_a^i)$ such that for every $\mu_a^i>\check{\mu}_a^i$, ${\rm C}_{(j,a')\cup l}^{(i,a)\cup m}$ is non-singular on $(\pi_a^i,\mu_a^i)$.
    \end{enumerate}
\end{theorem}

Theorem~\ref{equil_bund_theo} settles all the three assumptions with the context transformed from the unbiased central path to the equilibrium bundle.
\textbf{Starting point}: (i) shows that we can choose any policy as the starting policy $\pi_{a,0}^i$ simply by setting $\mu_{a,0}^i=\mu'\pi_{a,0}^i$, and then $(\pi_{a,0}^i,\mu_{a,0}^i)$ is sufficiently close to the equilibrium bundle when $\mu'$ is sufficiently large.
\textbf{Differentiability}: (ii) shows that the equilibrium bundle is differentiable as long as the coefficient matrix ${\rm C}_{(j,a')\cup l}^{(i,a)\cup m}$ is \textbf{non-singular} according to the implicit function theorem, and the differential $(d\pi_{a'}^j/\pi_{a'}^j)/(d\mu_{a''}^k/\mu_{a''}^k)$ points out the direction to update along the equilibrium bundle as $\mu_a^i$ decreases.
\textbf{Convexity}: (iii) shows that unbiased barrier problem \eqref{ubarr_equ} is strictly locally convex as long as the coefficient matrix ${\rm C}_{(j,a')\cup l}^{(i,a)\cup m}$ is \textbf{non-singular} since the only zero gradient point in the neighborhood is where $\pi_a^i=\hat{\pi}_a^i$.
Finally, (iv) shows that the point where ${\rm C}_{(j,a')\cup l}^{(i,a)\cup m}$ is \textbf{non-singular} can always be easily found.

A point where ${\rm C}_{(j,a')\cup l}^{(i,a)\cup m}$ is singular is a singular point of the equilibrium bundle as defined in Definition~\ref{equil_bund_def}.
Theorem~\ref{equil_bund_theo} (ii) and (iii) and the experiments show that singular points may block the iteration, and thus they must be avoided.
(iv) shows that this is always possible simply by adding $\beta^i\circ\pi_a^i$ to $\mu_a^i$ for a sufficiently large $\beta^i$.
Normally, it is harmless to pick any $\beta^i$ to perform the addition at any time, but note that the differential $(d\pi_{a'}^j/\pi_{a'}^j)/(d\mu_{a''}^k/\mu_{a''}^k)$ approaches identity matrix $I$ as $\mu_a^i\to \infty$.
In other words, on a point $(\pi_a^i, \mu_a^i)\in E$ where $\mu_a^i$ is too large, the equilibrium bundle is flat and the convergence is slow.
Considering all the above, we give the canonical section descent \eqref{canosect_formula} to update $\mu_a^i$, where $\eta^i\in [0,1)^i$ is the step length that is small enough, and $\beta^i\circ\pi_a^i$ is the term for singular point avoidance.
\begin{equation}
    \label{canosect_formula}
    \mu_{a,t+1}^i=(1-\eta_t^i)\circ\mu_{a,t}^i+\beta_t^i\circ\pi_{a,t}^i
\end{equation}

In summary, we've transformed the primal-dual unbiased interior point method into a \textbf{line search on the equilibrium bundle},
where \textbf{equilibrium bundle} is the structured solution space of unbiased KKT conditions \eqref{ukkt_equ},
the \textbf{first iteration level} is to update onto the equilibrium bundle by projected gradient descent \eqref{projgrad_formula},
the \textbf{second iteration level} is to hop across the fibers of the equilibrium bundle by canonical section descent \eqref{canosect_formula} and differential $(d\pi_{a'}^j/\pi_{a'}^j)/(d\mu_{a''}^k/\mu_{a''}^k)$,
and the line search leads to a Nash equilibrium as the canonical section $\bar{\mu}_a^i(\pi_a^i)$ reduces to $0$.
There is also a method to search globally for the entire set of Nash equilibria using the canonical section.
Specifically, we can sample policies in the policy space and calculate the canonical sections $\bar{\mu}_a^i(\pi_a^i)$ on them,
then we can set the samples where $\bar{\mu}_a^i(\pi_a^i)$ are close to $0$ as the starting points and approximate Nash equilibria near them.
The method is given by Algorithm~\ref{algo} as a single-state degeneration.

\begin{theorem}
    \label{odd_thm}
    The following statements about unbiased KKT conditions \eqref{ukkt_equ} hold.
    \begin{enumerate}
        \item Unbiased KKT conditions \eqref{ukkt_equ} is analytic with respect to both $(\pi_a^i,v^i)$ and $\mu_a^i$.
        \item {\rm (Oddness theorem)} For every $\mu_a^i$, there are almost always an odd number of solutions of unbiased KKT conditions \eqref{ukkt_equ}.
    \end{enumerate}
\end{theorem}

Analytic functions are functions that can be locally expand to a convergent power series, and being analytic is a property even stronger than being infinitely differentiable.
Theorem~\ref{odd_thm}~(i) is based on the fact that unbiased KKT conditions \eqref{ukkt_equ} derive a polynomial function that maps $(\pi_a^i,v^i)$ to $\mu_a^i$.
First, polynomial functions are analytic on their whole domain.
Second, the function with respect to $\mu_a^i$ is the inverse of a polynomial function that also satisfies the Newton-Puiseux theorem, such that it is analytic in the neighborhood of every point on it.
Furthermore, the singular points of the equilibrium bundle are all multiple roots of this polynomial function,
which are called critical points in terms of the polynomial, or algebraic branch points in terms of the inverse function of the polynomial.

Theorem~\ref{odd_thm}~(ii) gives an extension of the oddness theorem of Nash equilibria\cite{oddness2}, and it is implied by Newton-Puiseux theorem and that almost all points given by a polynomial equation are non-singular.
For any algebraic curve ${\rm AC}=\{(\pi_a^i,\hat{\mu}_a^i\mu')\in E|\mu'> \mu''\}$ given by $\hat{\mu}_a^i>0$ and $\mu''\geq 0$, denoting ${\rm EP}$ as the set of points on ${\rm AC}$ as $\mu'\to \mu''$.
First, there is exactly one point in ${\rm EP}$ that is connected by a branch of ${\rm AC}$ to the unique starting point of ${\rm AC}$ stated in Theorem~\ref{equil_bund_theo} (i) as $\mu'\to \infty$,
and all the other points in ${\rm EP}$ are connected in pairs by the rest branches of ${\rm AC}$, and thus $|{\rm EP}|$ is odd.
Second, ${\rm EP}$ equals to $\{(\pi_a^i,\hat{\mu}_a^i\mu'')\in E\}$ if every point in ${\rm EP}$ is non-singular, and this is almost always the case,
where $\{(\pi_a^i,\hat{\mu}_a^i\mu'')\in E\}$ is the entire set of solutions of \eqref{ukkt_equ} at $\hat{\mu}_a^i\mu''$,

\section{Dynamic programming on dynamic games}\label{dp_sec}
\subsection{Policy cone and best response cone}\label{dp_sec1}

In this section, we deal with the subproblem with respect to $V_s^i$ of problem \eqref{regmin_problem}, which is the dynamic programming problem \eqref{dp_problem}.
Solving dynamic programming problem \eqref{dp_problem} is equivalently making the convergence point of dynamic programming to lie within the constraint of problem \eqref{dp_problem}.
We first introduce two dynamic programming operators, along with the corresponding policy cone and best response cone,
such that the convergence point of the dynamic programming operator is the apex of the policy cone, and the constraint region of problem \eqref{dp_problem} is the best response cone.
Then we give Theorem~\ref{cone_dp} about the iterative properties of the dynamic programming operator within the policy cone,
and Theorem~\ref{cone_equil} about the equivalent conditions for the apex of the policy cone to be in the best response cone.
These two theorems are then used in the next subsection to construct the algorithm solving dynamic programming problem \eqref{dp_problem}.
\begin{equation}
    \label{dp_problem}
    \begin{aligned}
        \min_{V_s^i}\quad  & V_s^i-\pi_A^{si}(u_A^{si}+\gamma T_{s'A}^s V_{s'}^i)        \\
        \textrm{s.t.}\quad & V_s^i\geq\pi_{Aa}^{si-}(u_A^{si}+\gamma T_{s'A}^s V_{s'}^i) \\
    \end{aligned}
\end{equation}

We define two dynamic programming operators $D_\pi:\mathcal{V} \to \mathcal{V}$ and $\hat{D}_\pi:\mathcal{V} \to \mathcal{V}$ based on the objective and constraint of problem \eqref{dp_problem},
both of which are maps parameterized by policy $\pi_a^{si}$ such that
$$D_\pi(V_s^i)=\pi_A^s(u_A^{si}+\gamma T_{s'A}^s V_{s'}^i),\quad\hat{D}_\pi(V_s^i)=\max_a \pi_{Aa}^{si-}(u_A^{si}+\gamma T_{s'A}^s V_{s'}^i).$$
In single-player case, value iteration, which uses the Bellman operator $\max_a(u_a^s+\gamma T_{s'a}^s V_{s'})$ to iterate value function $V_s$, is a polynomial-time exact algorithm for computing optimal policies of MDPs.
But the similar operator generally does not converge in the multi-player dynamic games, which is explained in the next subsection.
The convergence of Bellman operator is contraction mapping convergence, but it is possible to settle for monotone convergence using $D_\pi$ and $\hat{D}_\pi$ through the policy cone and the best response cone we next define.
Contraction mapping convergence can be considered as a convergence in every direction, and monotone convergence can be considered as a convergence in only a certain direction, and we later show that the direction is $\mathbf{1}_s$.

\begin{figure}
    \centering
    \includegraphics[width=0.9\textwidth]{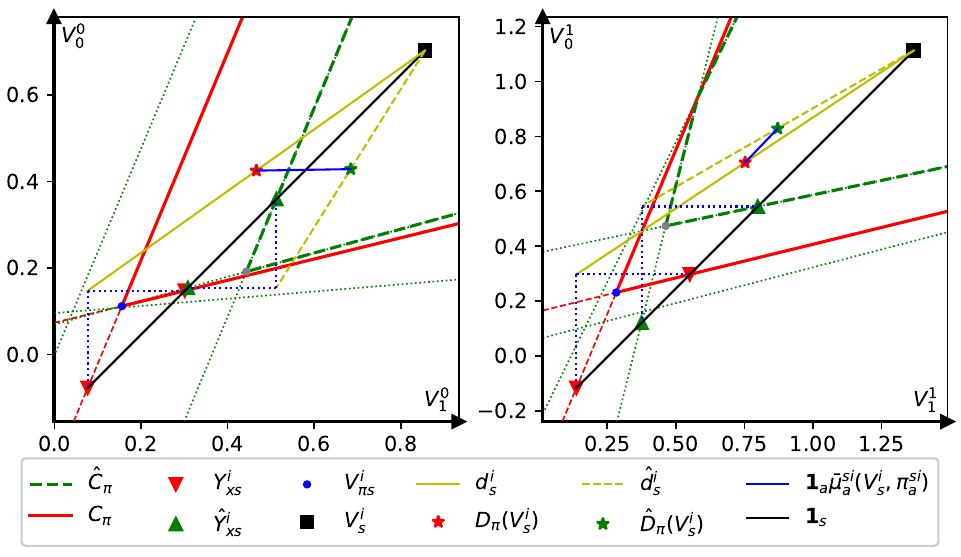}
    \caption{
    Graph of policy cone.
    This figure is plotted with a dynamic game where $N=\mathcal{S}=\mathcal{A}=\{0,1\}$.
    The graph is based on the value function space $\mathcal{V}$, where the two axes represent two state indices $s\in\{0,1\}$ of $V_s^i$, and the two subfigures represent two player indices $i\in\{0,1\}$ of $V_s^i$.
    As Proposition~\ref{cone_prop} shows, $C_\pi$ is a set of hyperplane-surrounded cone-shaped regions indexed by $i\in N$, with $V_{\pi s}^i$ being its apexes, and with $\hat{C}_\pi$ contained in it.
    $\mathbf{1}_s$ is the monotone convergence direction, which induces unique pairs $(Y_{xs}^i,d_x^i)$ and $(\hat{Y}_{xs}^i,\hat{d}_x^i)$, and satisfies that $V_s^i+m^i\mathbf{1}_s$ lies in $\hat{C}_\pi$ for any $V_s^i$ and sufficiently large $m^i$.
    Theorem~\ref{cone_dp} states that the residuals satisfy $V_s^i-D_\pi(V_s^i)=(1-\gamma)d_s^i$ and $V_s^i-\hat{D}_\pi(V_s^i)=(1-\gamma)\hat{d}_s^i$.
    Theorem~\ref{cone_equil} states that $\pi_a^{si}(x)$ is a Nash equilibrium if and only if the corresponding pair of $Y_{xs}^i$ and $\hat{Y}_{xs}^i$ coincide, and $\pi_a^{si}$ is a perfect equilibrium if and only if $V_{\pi s}^i$ lie in $\hat{C}_\pi$, in which case $Y_{xs}^i$, $\hat{Y}_{xs}^i$, and $V_{\pi s}^i$ all coincide.
    Equation \eqref{cano_dpres} shows that the relation between the canonical section and the two dynamic programming operators is $\mathbf{1}_a\bar{\mu}_a^{si}(V_s^i,\pi_a^{si})=\hat{D}_\pi(V_s^i)-D_\pi(V_s^i)$.
    }
    \label{cone_graph}
\end{figure}

\begin{definition}[Policy cone and best response cone]
    \label{cone_def}
    Let $\pi_a^{si}$ be a policy in dynamic game $\Gamma$.
    A policy cone $C_\pi$ and a best response cone $\hat{C}_\pi$ are regions in value function space $\mathcal{V}$ such that
    $$C_\pi=\left\{V_s^i\in\mathcal{V}|V_s^i\geq D_\pi(V_s^i) \right\},\quad\hat{C}_\pi=\left\{V_s^i\in\mathcal{V}|V_s^i\geq \hat{D}_\pi(V_s^i) \right\}.$$
\end{definition}

\begin{proposition}
    \label{cone_prop}
    Let $\pi_a^{si}$ be a policy in dynamic game $\Gamma$, and then the following properties hold.
    \begin{enumerate}
        \item There is a unique value function $V_{\pi s}^i$ that satisfies
              \begin{equation}
                  \label{policy_value}
                  V_{\pi s}^i=D_\pi(V_{\pi s}^i).
              \end{equation}
              In addition, $V_{\pi s}^i\in C_\pi$, and for any $V_s^i\in C_\pi$ there is $V_s^i\geq V_{\pi s}^i$.
        \item For any $V_s^i \in \mathcal{V} $ and $x \in \mathcal{S} $, there exists a unique pair $(Y_{xs}^i,d_x^i)\in\mathcal{V}\times\mathbb{R}^N$,
              as well as a unique pair $(\hat{Y}_{xs}^i,\hat{d}_x^i)\in\mathcal{V}\times\mathbb{R}^N$, such that
              \begin{subequations}
                  \begin{equation}
                      \label{cone_boundary}
                      \left(Y_{xs}^i-D_\pi(Y_{xs}^i)\right)(x)=0 \quad \wedge \quad V_s^i-Y_{xs}^i=d_x^i\mathbf{1}_s,
                  \end{equation}
                  \begin{equation}
                      \label{coneb_boundary}
                      \left(\hat{Y}_{xs}^i-\hat{D}_\pi(\hat{Y}_{xs}^i)\right)(x)=0 \quad \wedge \quad V_s^i-\hat{Y}_{xs}^i=\hat{d}_x^i\mathbf{1}_s.
                  \end{equation}
              \end{subequations}
        \item For any $V_s^i \in \mathcal{V} $, there exists an $M^i>0$ such that for any $m^i>M^i$, $V_s^i+m^i\mathbf{1}_s \in \hat{C}_\pi\subseteq C_\pi$.
    \end{enumerate}
\end{proposition}

According to Definition~\ref{cone_def} and Proposition~\ref{cone_prop} (i), policy cone $C_\pi$ is a set of cone-shaped regions indexed by $i\in N$,
each surrounded by $\left\lvert\mathcal{S}\right\rvert$ hyperplanes in $\left\lvert\mathcal{S}\right\rvert$-dimensional value function space $\mathcal{V}$, with $V_{\pi s}^i$ being a set of their apexes,
and best response cone $\hat{C}_\pi$ is a subset of policy cone $C_\pi$ for every $i\in N$, with each subset being an intersect of $\left\lvert\mathcal{A}\right\rvert$ cone-shaped regions.
Proposition~\ref{cone_prop} (ii) and (iii) show that both bottoms of $C_\pi$ and $\hat{C}_\pi$ expand towards infinity along $\mathbf{1}_s$ such that $V_s^i+m^i\mathbf{1}_s$ always lies in both of them for sufficiently large $m^i$.
And there is always a unique intersection $Y_{xs}^i$ between the line in direction $\mathbf{1}_s$ passing through $V_s^i \in \mathcal{V} $ and each hyperplane of $C_\pi$ indexed by $x$ and $i$,
and the same goes for $\hat{Y}_{xs}^i$ and $\hat{C}_\pi$.

\begin{theorem}[Iterative properties]
    \label{cone_dp}
    Let $\pi_a^{si}$ be a policy in dynamic game $\Gamma$, and then the following properties of dynamic programming operator $D_\pi$ hold.
    \begin{enumerate}
        \item $V_s^i\geq D_\pi(V_s^i)$ if and only if $V_s^i \in C_\pi$.
        \item For any $V_s^i \in C_\pi$, $D_\pi(V_s^i) \in C_\pi$.
        \item For any $V_s^i \in C_\pi$, if $D_\pi(V_s^i)=\hat{D}_\pi(V_s^i)$, then $D_\pi(V_s^i) \in \hat{C}_\pi$.
        \item For any $V_s^i \in C_\pi$, the residual
              \begin{subequations}
                  \begin{equation}
                      V_s^i-D_\pi(V_s^i)=(1-\gamma)d_s^i,
                  \end{equation}
                  \begin{equation}
                      V_s^i-\hat{D}_\pi(V_s^i)=(1-\gamma)\hat{d}_s^i,
                  \end{equation}
              \end{subequations}
              where $d_s^i$ and $\hat{d}_s^i$ are given by formula \eqref{cone_boundary} and \eqref{coneb_boundary} respectively.
        \item Let $V_s^i$ iterates by $V_{s,k+1}^i=D_\pi(V_{s,k}^i+m^i\mathbf{1}_s)$, where $m^i>0$ is a constant, and initial value function $V_{s,0}^i\in C_\pi$.
              Then
              \begin{equation}
                  \label{conv_to_apex}
                  \lim_{k \to \infty} V_{s,k}^i=V_{\pi s}^i+\frac{\gamma}{1-\gamma}m^i\mathbf{1}_s \quad \wedge \quad \lim_{k \to \infty} \left(V_{s,k}^i-D_\pi(V_{s,k}^i)\right)=\gamma m^i\mathbf{1}_s.
              \end{equation}
    \end{enumerate}
\end{theorem}

Theorem~\ref{cone_dp} uses policy cone $C_\pi$ and best response cone $\hat{C}_\pi$ to describe iterative properties of dynamic programming operator $D_\pi$ for fixed policy $\pi_a^{si}$.
(i), (ii), and (iii) suggest that $C_\pi$ is the monotonic and closed domain for $D_\pi$,
and $\hat{C}_\pi$ is closed if $\pi_a^{si}(x)$ is a Nash equilibrium for every state $x$.
That is, for any iteration starting from value function $V_s^i$ within $C_\pi$, the value function decreases monotonically as iterates and never leaves $C_\pi$,
which means the iteration converges to the apex $V_{\pi s}^i$ of $C_\pi$ by the monotone convergence theorem.
(iv) illustrates that iteration residual $V_s^i-D_\pi(V_s^i)$ of $D_\pi$ can be expressed by the distance $d_x^i$ pairing with the unique intersection $Y_{xs}^i$ described in Proposition~\ref{cone_prop} (ii),
where the set of distances $d_x^i$ indexed by state $x$ is used as a single vector $d_s^i$,
and the same goes for $\hat{D}_\pi$, $\hat{d}_x^i$, and $\hat{Y}_{xs}^i$.
In addition to value function $V_s^i$ iteratively converging to apex $V_{\pi s}^i$, (v) shows that if a scaled $\mathbf{1}_s$ is added in every iteration,
not only does $V_s^i$ converge to $V_{\pi s}^i$ adding a scaled $\mathbf{1}_s$, but $V_s^i-D_\pi(V_s^i)$ also converge to a scaled $\mathbf{1}_s$,
indicating that $\mathbf{1}_s$ is the direction of the monotone convergence of $D_\pi$.

\begin{theorem}[Equilibrium conditions]
    \label{cone_equil}
    The following equivalent conditions of equilibrium hold.
    \begin{enumerate}
        \item $\pi_a^{si}$ is a perfect equilibrium if and only if $V_{\pi s}^i \in \hat{C}_\pi$, and if and only if $V_{\pi s}^i$ is the global optimal point of dynamic programming problem \eqref{dp_problem}.
        \item Given value function $V_s^i$, for any $x \in \mathcal{S}$,
              $\pi_a^{si}(x)$ is a Nash equilibrium of the static game with utility function $(u_A^{si}+\gamma T_{s'A}^s V_{s'}^i)(x)$ if and only if $Y_{xs}^i=\hat{Y}_{xs}^i$,
              where $Y_{xs}^i$ and $\hat{Y}_{xs}^i$ are given by formula \eqref{cone_boundary} and \eqref{coneb_boundary} respectively.
    \end{enumerate}
\end{theorem}

Theorem~\ref{cone_equil} uses $C_\pi$ and $\hat{C}_\pi$ to describe equivalent conditions of equilibria.
(ii) shows that policy $\pi_a^{si}(x)$ is a Nash equilibrium for state $x$ if and only if the unique intersects $Y_{xs}^i$ and $\hat{Y}_{xs}^i$ coincide for state $x$.
(i) shows that policy $\pi_a^{si}$ is a perfect equilibrium if and only if the apex $V_{\pi s}^i$ of $C_\pi$ is in $\hat{C}_\pi$.
Note that when $\pi_a^{si}$ is a perfect equilibrium, $Y_{xs}^i$, $\hat{Y}_{xs}^i$, and $V_{\pi s}^i$ all coincide,
which is exactly the definition of perfect equilibrium such that $V_{\pi s}^i$ is the value function of $\pi_a^{si}$ and $\pi_a^{si}(x)$ is a Nash equilibrium for every state $x$.
In addition, Theorem~\ref{cone_equil} also shows the relation between dynamic programming problem \eqref{dp_problem} and perfect equilibrium.

\subsection{Iteration in policy cone}\label{dp_sec2}
Theorem~\ref{cone_dp} shows that for fixed policy, the iteration of dynamic programming operator $D_\pi$ in the policy cone converges monotonically to the apex,
and Theorem~\ref{cone_equil} shows the equivalent condition for the apex of policy cone to be a perfect equilibrium value function.
We next study the possibility for iterative methods constructed by $D_\pi$ to converge to a perfect equilibrium value function.
First, we give Proposition~\ref{naive_iter} to explain why single-player dynamic programming cannot be simply generalized to multi-player dynamic games,
then we give Theorem~\ref{cone_iter} that asserts the sufficient and necessary condition for $D_\pi$ to converge to a perfect equilibrium value function.

\begin{proposition}
    \label{naive_iter}
    Let $V_s^i$ iterates by $V_{s,k+1}^i=D(V_{s,k}^i):=D_{\pi_k}(V_{s,k}^i)$, where $\pi_{a,k}^{si}$ satisfies $D_{\pi_k}(V_{s,k}^i)=\hat{D}_{\pi_k}(V_{s,k}^i)$.
    If $V_{s,0}^i\in C_{\pi_0}$, then the following statements satisfy $(i)\Rightarrow(ii)\Rightarrow(iii)\Leftrightarrow(iv)$.
    \begin{enumerate}
        \item $V_s^i\leq Y_s^i\rightarrow D(V_s^i)\leq D(Y_s^i)$ for every $V_s^i,Y_s^i\in O(\hat{V}_s^i)$, where $O(\hat{V}_s^i)$ is a set such that $V_{s,k}^i\in O(\hat{V}_s^i)$ for all $k$.
        \item $V_{s,k}^i$ converges to a perfect equilibrium value function by contraction mapping $D$, and $V_{s,k}^i\in C_{\pi_k}$ for all $k$.
        \item $V_{s,k}^i$ monotonically converges to a perfect equilibrium value function.
        \item $V_{s,k+1}^i\leq V_{s,k}^i\rightarrow D(V_{s,k+1}^i)\leq D(V_{s,k}^i)$ for all $k$.
    \end{enumerate}
\end{proposition}

Using Theorem~\ref{cone_dp} (iii), it can be inferred that Proposition~\ref{naive_iter} (iv) implies $V_{s,k+1}^i\in \hat{C}_{\pi_k}\rightarrow V_{s,k+1}^i \in C_{\pi_{k+1}}$.
It can be verified on the graph of policy cone that this implication formula generally does not hold, since $\pi_k$ and $\pi_{k+1}$ generally does not have a strong relation.
Consequently, none of the four statements hold in dynamic games.
In particular, neither does the first half of (ii) that operator $D$ is a contraction mapping hold on its own, since previous research already shows that operator $D$ fails to converge in dynamic games\cite{mdp_game,nash_q}.
However, as the single-player degeneration of operator $D$, the Bellman operator $D$ in MDPs satisfies (i), which can be used to prove the first half of (ii) that Bellman operator $D$ is a contraction mapping.
Furthermore, if there is also $V_{s,0}^i\in C_{\pi_0}$ for Bellman operator $D$, then all the four statements hold as Proposition~\ref{naive_iter} shows, where perfect equilibrium degenerates to optimal policy.
In summary, Bellman operator $D$ can achieve both monotone convergence and contraction mapping convergence in single-player case, but can achieve neither of the convergence in multi-player case.
Proposition~\ref{naive_iter} is not necessary to construct our iterative approximation method in this paper, however, it provides a certain perspective why the Bellman operator and value iteration cannot be simply generalized to dynamic games.

Then we construct an iteration of dynamic programming operator $D_\pi$ that converges to a perfect equilibrium value function.
Note that not only is $V_s^i$ bounded within $C_\pi$ under the iteration of $D_\pi$ for fixed $\pi_a^{si}$, but the apex $V_{\pi s}^i$ of $C_\pi$ is also bounded with respect to the variable $\pi_a^{si}$.
This means that if $V_{s,k}^i \in C_{\pi_k}$ for all $k\in \mathbb{N}$ and $V_{s,k+1}^i=D_{\pi_k}(V_{s,k}^i)$, then $\{V_{s,k}^i\}_{k\in \mathbb{N} }$ still satisfies the monotone convergence property.
This allows us to construct the cone interior convergence conditions.
First, we add a large enough $m_k^i\mathbf{1}_s$ to $V_{s,k}^i$ at every step to keep $V_{s,k}^i+m_k^i\mathbf{1}_s$ in the policy cone $C_{\pi_k}$, and let $m_k^i$ converge to $0$, so that monotone convergence property is maintained and $V_{s,k}^i$ converges to the apex of $C_{\pi_k}$.
Second, the apex that $V_{s,k}^i$ converges to has to be in the best response cone $\hat{C}_{\pi_k}$ for the apex to be a perfect equilibrium value function, thus we further require $V_{s,k}^i+m_k^i\mathbf{1}_s$ be in the best repsonse cone $\hat{C}_{\pi_k}$.
Summarizing these gives us Theorem~\ref{cone_iter}.

\begin{theorem}[Cone interior convergence conditions]
    \label{cone_iter}
    Let $V_s^i$ iterates by
    \begin{equation}
        \label{dp_formula}
        V_{s,k+1}^i=D_{\pi_k}(V_{s,k}^i+m_k^i\mathbf{1}_s),
    \end{equation}
    where $\{\pi_{a,k}^{si}\}_{k\in \mathbb{N} }$ is a sequence of policies,
    and $\{m_k^i\}_{k\in \mathbb{N} }$ is a sequence sufficiently large such that $V_{s,k}^i+m_k^i\mathbf{1}_s \in \hat{C}_{\pi_k}$ for all $k\in \mathbb{N}$, which always exists.
    Then, $V_{s,k}^i$ converges to a perfect equilibrium value function if and only if $\lim_{k \to \infty}m_k^i=0$.
\end{theorem}

Theorem~\ref{cone_iter} points out a dynamic programming method that iteratively converges to a perfect equilibrium sufficiently and necessarily, where $m_k^i$ and $\pi_{a,k}^{si}$ are required in each iteration step.
First, for simplification, we let $m_k^i\equiv m^i$, thus $m^i$ only needs to be large enough so that $V_{s,k}^i+m^i\mathbf{1}_s \in \hat{C}_{\pi_k}$,
and $V_{s,k}^i$ converges to the value function $V_{\pi s}^i$ plus a scaled $\mathbf{1}_s$ as stated in Theorem~\ref{cone_dp} (v).
Second, we ensure that $\pi_{a,k}^{si}(x)$ converges to a Nash equilibrium for every state $x$.
Then we've transformed the problem of approximating a perfect equilibrium to the problem of approximating a set of Nash equilibria, which is solved in the last section.

\section{FPTAS for perfect equilibria of dynamic games}\label{hybrid_sec}
\subsection{Line search on the equilibrium bundle of dynamic games}
In this section, we combine the two methods in the last two sections to obtain a hybrid iteration of dynamic programming and interior point method that solves the regret minimization problem \eqref{regmin_problem} of dynamic games, and equivalently approximates any perfect equilibrium of any dynamic game.
First, we add the state index $s$ back to the tensors in the static game case, and set the utility $U_A^{si}=u_A^{si}+\gamma T_{s'A}^s V_{s'}^i$, then we have the following theorem.

\begin{theorem}
    \label{per_equil_thm}
    Let $\bar{\mu}_a^{si}(V_s^i,\pi_a^{si})$ be a family of canonical sections such that for each $x\in\mathcal{S}$,
    $\bar{\mu}_a^{si}(V_s^i,\pi_a^{si})(x)$ is the canonical section of the equilibrium bundle given by $U_A^{si}(x)$,
    where $U_A^{si}=u_A^{si}+\gamma T_{s'A}^s V_{s'}^i$.
    Then there is
    \begin{equation}
        \label{cano_dpres}
        \mathbf{1}_a\bar{\mu}_a^{si}(V_s^i,\pi_a^{si})=\hat{D}_\pi(V_s^i)-D_\pi(V_s^i),
    \end{equation}
    and the following statements are equivalent.
    \begin{enumerate}
        \item $\pi_a^{si}$ is a perfect equilibrium.
        \item $\bar{\mu}_a^{si}(V_{\pi s}^i,\pi_a^{si})=0$.
        \item $\hat{D}_\pi(V_{\pi s}^i)=D_\pi(V_{\pi s}^i)$.
        \item $(\pi_a^{si},V_{\pi s}^i)$ is a global optimal point of regret minimization problem \eqref{regmin_problem} of dynamic games.
    \end{enumerate}
\end{theorem}

Theorem~\ref{per_equil_thm} summarizes the equilibrium conditions given by the last two sections.
It first points out the connection between the canonical section and the two dynamic programming operators, and then shows the equivalent condition (ii) and (iii) of perfect equilibria based on this connection.
(iv) combines the equivalent conditions about the global optimal point of the two subproblems of problem \eqref{regmin_problem} given by Theorem~\ref{regmin_equiv} and Theorem~\ref{cone_equil} to obtain the equivalent condition about the global optimal point of problem \eqref{regmin_problem}.

Then, we combine the two methods in the last two sections to obtain a hybrid iteration, where $V_s^i$ iterates by dynamic programming \eqref{dp_formula}, $\pi_a^{si}$ iterates by projected gradient \eqref{projgrad_formula}, and $\mu_a^{si}$ iterates by canonical section descent \eqref{canosect_formula}.
However, note that projected gradient can update back onto the equilibrium bundle only after an infinitesimal deviation,
and there is a problem that $V_s^i$ moves more than an infinitesimal step in every iteration of the dynamic programming, except for the case where $V_s^i$ is nearly converged to $V_{\pi s}^i$ plus a scaled $\mathbf{1}_s$.
Thus second, we set $U_{\pi A}^{si}=u_A^{si}+\gamma T_{s'A}^s V_{\pi s'}^i$ related to $\pi_a^{si}$, which gives us the following results about dynamic games.

\ukktdy*
\equilbundl*

\begin{theorem}[Existence and oddness theorems]
    \label{exist_odd_thm}
    The following statements about unbiased KKT conditions \eqref{ukkt_dy_equ} of dynamic games hold.
    \begin{enumerate}
        \item For every $\mu_a^{si}$, there is at least one solution of \eqref{ukkt_dy_equ}.
        \item For every $\mu_a^{si}$, there are almost always an odd number of solutions of \eqref{ukkt_dy_equ}.
    \end{enumerate}
\end{theorem}

There are three facts about unbiased KKT conditions \eqref{ukkt_dy_equ} about dynamic games.
First, it derives a Brouwer function $\hat{\pi}_a^{si}=M(\mu_a^{si})(\pi_a^{si})$, such that the Brouwer's fixed point theorem applies, leading to the existence theorem.
Second, it derives a polynomial function from $\pi_a^{si}$ to $\mu_a^{si}$ by multiplying the non-zero determinant $|I_{s'}^s-\gamma T_{\pi s'}^s|$ to the second line, such that the Newton-Puiseux theorem applies, leading to the oddness theorem.
Third, its solution space can be structured as the equilibrium bundle, which is a fiber bundle that formalizes perfect equilibria as the zero points of its canonical section, and formalizes the hybrid iteration of dynamic programming and interior point method as a line search on it.

\begin{algorithm}
    \caption{A hybrid iteration of dynamic programming and interior point method}
    \label{algo}
    \begin{algorithmic}[1]
        \Require Dynamic game $\Gamma=(N,\mathcal{S},\mathcal{A},T,u,\gamma)$.
        \State (Optional) Sample in the policy space $\mathcal{P}$ and calculate canonical sections $\bar{\mu}_a^{si}(\pi_a^{si})$ to search for potential perfect equilibria globally.
        \Require Initial policy $\pi_a^{si}$, either random or around a potential perfect equilibrium.
        \State Set initial barrier parameter $\mu_a^{si}=\pi_a^{si}\mu'$ for sufficiently large $\mu'$, so that $\mu_a^{si}$ is around the fiber over $\pi_a^{si}$ by Theorem~\ref{equil_bund_theo} (i).
        \State Set initial value function $V_s^i=m^i\mathbf{1}_s$ for sufficiently large $m^i$, so that $V_s^i \in \hat{C}_\pi$ by Proposition~\ref{cone_prop} (iii).
        \Repeat
        \Repeat
        \State Set $U_A^{si}=u_A^{si}+\gamma T_{s'A}^s (V_{s'}^i+m^i\mathbf{1}_{s'})$ to apply dynamic programming \eqref{dp_formula}.
        \State Calculate $\pi_{Aaa'}^{sij-}U_A^{si}$ and $\pi_{Aa}^{si-}U_A^{si}$, either model-based or model-free.
        \State Compute $v_s^i$ by solving equation \eqref{ukkt_equ1} as stated in Theorem~\ref{ipm_theo} (i).
        \State Calculate regret $r_a^{si}=v_s^i-\pi_{Aa}^{si-}U_A^{si}$.
        \State Calculate dual policy $\hat{\pi}_a^{si}=\mu_a^{si}/r_a^{si}$ and dual regret $\hat{r}_a^{si}=\mu_a^{si}/\pi_a^{si}$.
        \State Calculate projected gradient ${\rm pg}_a^{si}$ using equation \eqref{projgrad_formula}.
        \State Calculate residual $dV_s^i:=D_\pi(V_s^i+m^i\mathbf{1}_s)-V_s^i=v_s^i-\pi_a^{si}r_a^{si}-V_s^i$ of dynamic programming \eqref{dp_formula}.
        \State Update $\pi_a^{si}$ by ${\rm pg}_a^{si}$, and update $V_s^i$ by $dV_s^i$.
        \Until{$\pi_a^{si}-\hat{\pi}_a^{si}$, $r_a^{si}-\hat{r}_a^{si}$, and ${\rm Angle}(dV_s^i,\mathbf{1}_s)$ all converge to $0$, so that $(\pi_a^{si},\mu_a^{si})$ is on the equilibrium bundle.}
        \State Compute differential $(d\pi_{a'}^{sj}/\pi_{a'}^{sj})/(d\mu_{a''}^{sk}/\mu_{a''}^{sk})$ by solving equation \eqref{tan_vec}.
        \State Calculate canonical section $\bar{\mu}_a^{si}(\pi_a^{si})=\pi_a^{si} \circ (\max_a \pi_{Aa}^{si-}U_{\pi A}^{si}-\pi_{Aa}^{si-}U_{\pi A}^{si})$.
        \State Update $\mu_a^{si}$ using canonical section descent \eqref{canosect_formula} to hop to another fiber in the neighborhood and move along the fiber to avoid potential singular points,
        and update $\pi_a^{si}$ along differential $(d\pi_{a'}^{sj}/\pi_{a'}^{sj})/(d\mu_{a''}^{sk}/\mu_{a''}^{sk})$.
        \Until{$\bar{\mu}_a^{si}(\pi_a^{si})$ converges to $0$, so that the fiber over a perfect equilibrium is reached.}
    \end{algorithmic}
\end{algorithm}

\begin{figure}
    \centering
    \includegraphics[width=0.9\textwidth]{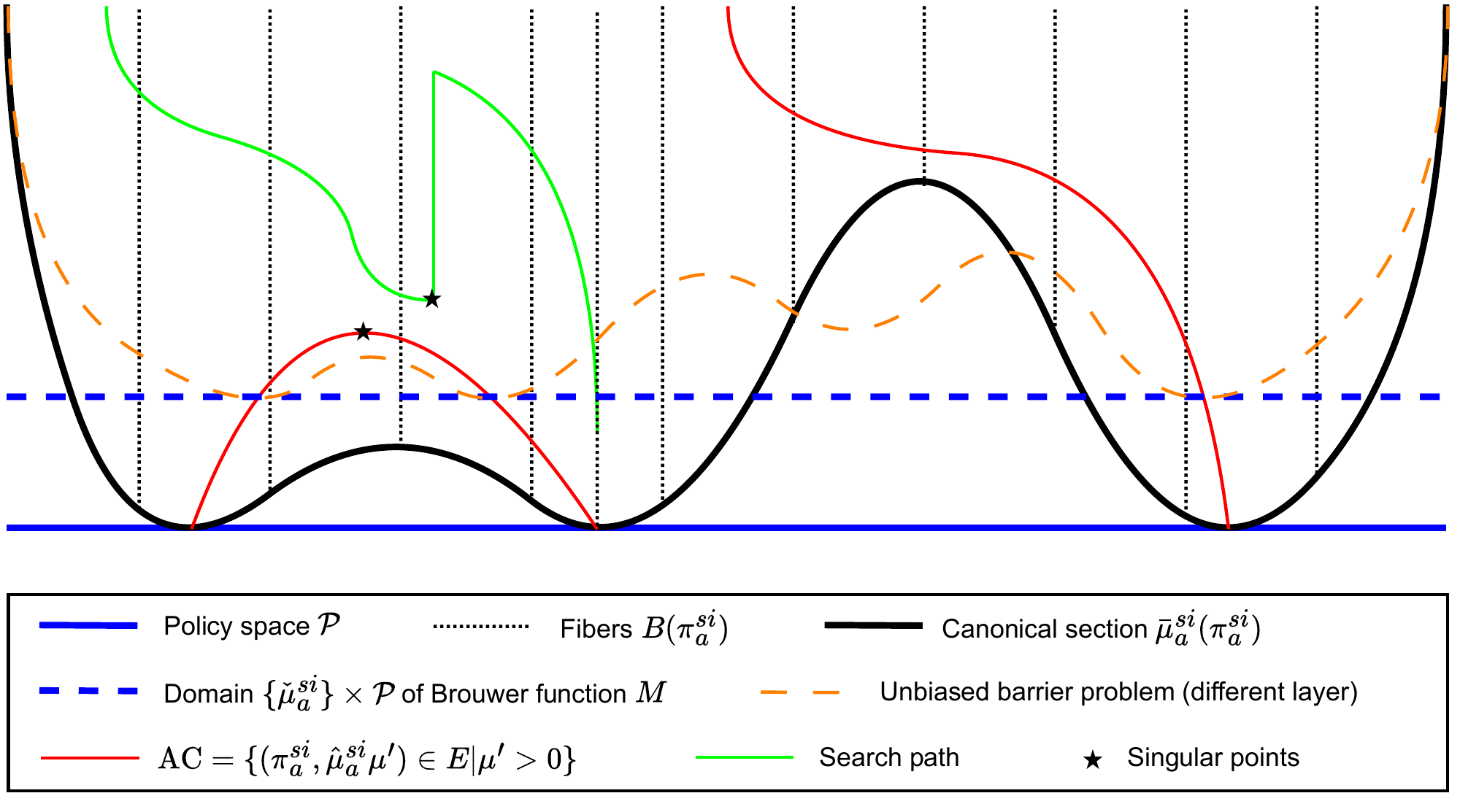}
    \caption{
    Sketch graph of the equilibrium bundle.
    This sketch graph is based on the joint space $\mathcal{P}\times\{\mu_a^{si}|\mu_a^{si}\geq 0\}$ of policy and barrier parameter.
    First, equilibrium bundle $E$ consists of the disjoint union of fibers $\{\pi_a^{si}\}\times B(\pi_a^{si})$ over each $\pi_a^{si}\in\mathcal{P}$, where each fiber $B(\pi_a^{si})$ is an affine subspace with the canonical section $\bar{\mu}_a^{si}(\pi_a^{si})$ being its least element, and perfect equilibria are zero points of the map $\bar{\mu}_a^{si}$.
    Second, Brouwer function $\hat{\pi}_a^{si}=M(\mu_a^{si})(\pi_a^{si})$ rearranges the points in the offset policy space $\mathcal{P}\times\{\mu_a^{si}\}$, where the fixed points are the intersections between $\mathcal{P}\times\{\mu_a^{si}\}$ and the equilibrium bundle $E$, and unbiased barrier problem depicts the approximation to those fixed points.
    Third, the polynomial function from $\pi_a^{si}$ to $\mu_a^{si}$ derives an algebraic curve ${\rm AC}$ that satisfies a parity argument, such that exactly one of its endpoints is connected with the starting point as $\mu'\to \infty$, and the rest of the endpoints are connected in pairs.
    In addition, the singular points of the equilibrium bundle are the multiple roots of the polynomial function.
    Finally, our method is a line search on the equilibrium bundle, which hops across the fibers to a zero point of the canonical section, and moves along a fiber to avoid singular points where the differential $d\pi_a^{si}/d\mu_{a'}^{sj}$ tends to infinity.
    }
    \label{equilbund_graph}
\end{figure}

Algorithm~\ref{algo} is still a \textbf{line search on the equilibrium bundle} that consists of two iteration levels,
the \textbf{first iteration level} is to update onto the equilibrium bundle by alternating the steps of dynamic programming \eqref{dp_formula} and projected gradient descent \eqref{projgrad_formula},
the \textbf{second iteration level} is to hop across the fibers of the equilibrium bundle by canonical section descent \eqref{canosect_formula} and differential $(d\pi_{a'}^{sj}/\pi_{a'}^{sj})/(d\mu_{a''}^{sk}/\mu_{a''}^{sk})$,
and the line search leads to a perfect equilibrium as the canonical section $\bar{\mu}_a^{si}(\pi_a^{si})$ reduces to $0$.
Then we show that Algorithm~\ref{algo} converges in polynomial time by showing the convergence rates of the three iteration formulas.

\begin{proposition}
    \label{conv_rate}
    The following statements about convergence rate hold.
    \begin{enumerate}
        \item $\lVert V_s^i-D_\pi(V_s^i+m^i\mathbf{1}_s)\rVert _\infty$ converge to $0$ linearly with a rate of $\gamma$ under the dynamic programming \eqref{dp_formula} for fixed $\pi_a^{si}$ and $m$.
        \item $\lVert \bar{\mu}_a^i(\pi_a^i)\rVert _\infty$ converge to $0$ linearly for some $\eta^i>0$ and $\beta^i\geq 0$ under the canonical section descent \eqref{canosect_formula}.
        \item $(\pi_a^i-\hat{\pi}_a^i)(r_a^i-\hat{r}_a^i)$ converge to $0$ sublinearly under the iteration of projected gradient \eqref{projgrad_formula} if the convergence point is non-singular.
    \end{enumerate}
\end{proposition}

In Proposition~\ref{conv_rate}, for a linearly converging iteration, the number of iterations needed to achieve given precision $\epsilon$ is $O(\log^m(1/\epsilon))$,
and projected gradient descent is known to converge sublinearly and require $O((1/\epsilon)^n)$ iterations to achieve given precision $\epsilon$.
Note also that the elementary operations of the three iterations are all tensor operations whose complexity is $O(\lvert N\rvert^o\lvert\mathcal{A}\rvert^p\lvert\mathcal{S}\rvert^q)$.
And Theorem~\ref{equil_bund_theo} (iv) assures that for every dynamic game, there are always non-singular points available for the convergence so that Proposition~\ref{conv_rate} applies.
Thus, the running time of Algorithm~\ref{algo} is polynomial with respect to $|N|$, $|\mathcal{A}|$, $|\mathcal{S}|$, and $1/\epsilon$.
In other words, Algorithm~\ref{algo} is an FPTAS that achieves a weak approximation of perfect equilibria, where the canonical section satisfies equation \eqref{eps_equil}, in fully polynomial time.
\begin{equation}
    \label{eps_equil}
    \bar{\mu}_a^{si}(\pi_a^{si})=\pi_a^{si} \circ \left(\max_a \pi_{Aa}^{si-}U_{\pi A}^{si}-\pi_{Aa}^{si-}U_{\pi A}^{si}\right)<\epsilon_a^{si}.
\end{equation}

Recall the two different approximations: weak approximation approximates to an $\epsilon$-equilibrium, and strong approximation approximates to an $\epsilon$-neighborhood of an exact equilibrium.
It directly follows that on the equilibrium bundle, a strong approximation is a weak approximation, while the opposite is not true, such that an $\epsilon$-equilibrium could be far from an exact equilibrium, depending on the nearby gradient of the canonical section,
which aligns with existing results\cite{different_appro,nash_strong_complex1}.
For any perfect equilibrium of any dynamic game, our method achieves a weak approximation in fully polynomial time, and the time complexity for our method to achieve a strong approximation depends on the gradient of the canonical section near the exact equilibrium.
Recall the complexity results that the weak approximation of Nash equilibria is PPAD-complete, and strong approximation of Nash equilibria with three or more players is FIXP-complete.
This implies PPAD=FP.

\subsection{Practical use}

It is possible to construct a MARL method through Algorithm~\ref{algo}, and the derived method would be free from non-stationarity and curse of multiagency since Algorithm~\ref{algo} guarantees convergence and has polynomial time complexity.
Algorithm~\ref{algo} fits into the actor-critic framework, where $\pi_a^{si}$ is the actor, $V_s^i$ is the critic, and $\bar{\mu}_a^{si}$ is the loss function.
In addition, the canonical section $\bar{\mu}_a^{si}$ can be used to guide the global search of the policy space, and the term $m_k^i\mathbf{1}_s$ in dynamic programming can be used to balance exploration and exploitation of the state space.

Note that the normal-form representation $U_A^{si}$ only involves in the calculation through $\pi_{Aa}^{si-}U_A^{si}$ and $\pi_{Aaa'}^{sij-}U_A^{si}$ in Algorithm~\ref{algo},
where the problem of computing $\pi_{Aa}^{si-}U_A^{si}$ and $\pi_{Aaa'}^{sij-}U_A^{si}$ is a variant of a problem called the expected utility problem that computes $\pi_A^s U_A^{si}$.
Thus, Algorithm~\ref{algo} actually works with any game representation, as long as the expected utility problem has an polynomial-time algorithm in that representation, such as those important classes of succint games like graphical games, action-graph games, and many others.
In addition, Algorithm~\ref{algo} even works in model-free cases, where $\pi_{Aa}^{si-}U_A^{si}$ and $\pi_{Aaa'}^{sij-}U_A^{si}$ are estimated using sampled data collected as an self-play agent interacts with a game instance.

Algorithm~\ref{algo} takes any dynamic game as input to output its perfect equilibrium, such as it takes a static game as a single-state dynamic game to produce its Nash equilibrium, and it takes an MDP as a single-player dynamic game to produce its optimal policy.
We implement the algorithm to take any size of dynamic game as input, and we animate the line search process for dynamic games with 2 players, 2 states, and 2 actions using Fig.~\ref{cone_graph}, Fig.~\ref{barrproblem_graph}, Fig.~\ref{kktcondition_graph}, and the iteration curve in Fig.~\ref{iter_curve}.
And we tested our method on 2000 randomly generated dynamic games of 3 players, 3 states, and 3 actions in experiment, and the iteration converges to a perfect equilibrium in every single case.

\begin{figure}
    \centering
    \includegraphics[width=0.9\textwidth]{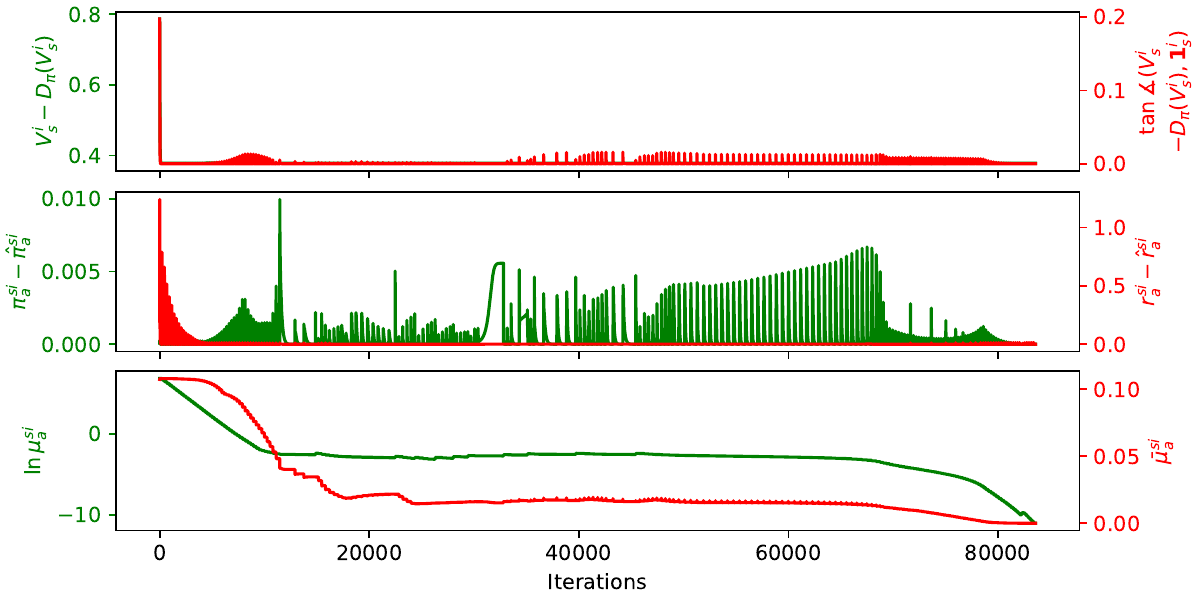}
    \caption{
    Iteration curve.
    The figure shows the convergence of the three iterations in Proposition~\ref{conv_rate}.
    ${\rm Angle}(V_s^i-D_\pi(V_s^i),\mathbf{1}_s^i)$, $(\pi_a^{si}-\hat{\pi}_a^{si},r_a^{si}-\hat{r}_a^{si})$, and $\bar{\mu}_a^{si}(\pi_a^{si})$ all converging to $0$ indicates that the convergence point is a perfect equilibrium.
    In particular, ${\rm Angle}(V_s^i-D_\pi(V_s^i),\mathbf{1}_s^i)$ and $(\pi_a^{si}-\hat{\pi}_a^{si},r_a^{si}-\hat{r}_a^{si})$ staying converged during the whole iteration indicates that the iteration is a line search on the equilibrium bundle,
    and $\mu_a^{si}$ not decreasing in the middle of the iteration is due to singular avoidance.
    }
    \label{iter_curve}
\end{figure}

\subsection{Tractability of PPAD}
PPAD is defined as the complexity class of all the problems that reduces to {\it End-Of-The-Line} in polynomial time, and it is believed to contain hard problems because {\it End-Of-The-Line} is seemingly intractable in polynomial time, yet our discovery implies PPAD=FP.
Thus, it is necessary that we explain how the reduction of our method can solve {\it End-Of-The-Line} in polynomial time.

We've introduced that {\it End-Of-The-Line} is believed to be intractable because it seems to let us follow a potentially exponentially long chain in directed graph $DG$ from a given source to a sink, only by using a polynomial-time computable function $f$ that outputs the predecessor and successor of an input vertex.
However, note that it is not necessary to follow the exponentially long chain or use function $f$ to jump over vertices one by one.
The reduction of our method to {\it End-Of-The-Line} consists of two independent steps: computing a static game whose $\epsilon$-approximate Nash equilibria correspond to unbalanced vertices of $DG$ in polynomial time, and computing an $\epsilon$-approximate Nash equilibrium in fully polynomial time.
The first step is given by the reduction from {\it End-Of-The-Line} to {\it Nash}\cite{nash_weak_complex1}, and the second step is our method.

The reduction of {\it End-Of-The-Line} to {\it Nash} consists of two steps that can be done in polynomial time.
First, construct a Brouwer function $F$ on the unit cube that is given by an arithmetic circuit that consists only of addition, multiplication and comparison, with the boolean circuit $f$ being its subcircuit,
such that every vertex of $DG$ is encoded in the unit cube, and the unbalanced vertices except the given one are exactly the $\epsilon$-approximate fixed points.
Second, simulate every operation in arithmetic circuit $F$ with a two-action static game to obtain a many-player two-action graphical game, and then simulate the graphical game with a three-player many-action static game,
such that the weights of certain actions on $\epsilon'$-approximate Nash equilibria are exactly the $\epsilon$-approximate fixed points.

After the static game is constructed, our method is used to compute an $\epsilon$-approximate Nash equilibria, and then certain components of the policy can be decoded to obtain an unbalanced vertex.
Our method is a line search on the equilibrium bundle, the search path lies in the total space, and $DG$ is encoded in the base space, where the dimension of the total space is twice that of the base space.
The line search originates from the fiber over the point encoding the given unbalanced vertex, and eventually reaches the fiber over an $\epsilon$-approximate Nash equilibria encoding another unbalanced vertex.

In summary, there are two facts that make {\it End-Of-The-Line} solvable in polynomial time.
First, the polynomial-time computable function $f$ is not used to jump over vertices one by one, instead, it is used to construct a static game in polynomial time, such that the directed graph is encoded in its policy space and the unbalanced vertices are encoded as Nash equilibria.
Second, the search path does not follow the exponentially long chain, instead, it is completely another path that lies in a space with twice the dimension of the policy space where the chain is encoded.

\section{Conclusion}
In this paper, we aim to deal with fully observable dynamic games to find a polynomial-time algorithm to approximate perfect equilibria.
We first formalize the game equilibrium problem as an optimization problem, which splits into two subproblems with respect to policy and value function.
The subproblem with respect to policy is equivalent to the Nash equilibrium problem of static games, for which we introduce the unbiased barrier problem and unbiased KKT conditions to make the interior point method to approximate Nash equilibria of static games.
The subproblem with respect to value function is equivalent to the convergence problem of dynamic programming in dynamic games, for which we introduce the policy cone to give the sufficient and necessary condition for dynamic programming to converge to perfect equilibria of dynamic games.
Finally, combining the two sections of results, we introduce the equilibrium bundle of dynamic games, such that it formalize the perfect equilibria as the zero points of its canonical section,
and it formalize a hybrid iteration of dynamic programming and interior point method as a line search on it.
The geometric properties of the equilibrium bundle allows us to give the existence and oddness theorems as an extension of those of Nash equilibria.
In addition, the equilibrium bundle, unbiased barrier problem, unbiased KKT conditions, and Brouwer function all lie in the joint space of policy and barrier parameter with certain geometric structure.
The hybrid iteration approximate any perfect equilibrium of any dynamic game, it achieves a weak approximation in fully polynomial time, the time complexity for it to achieve a strong approximation depends on the nearby gradient of the canonical section.
This makes the method an FPTAS for the PPAD-complete weak approximation problem of game equilibria, implying PPAD=FP.
In experiments, the line search process is animated, and the method is tested on 2000 randomly generated dynamic games where it converges to a perfect equilibrium in every single case.

\section*{Acknowledgments}
\textbf{Funding:}
This work was supported by the National Key R\&D Program of China (2022YFB4701400/4701402), National Natural Science Foundation of China (No. U21B6002, 62203260, 92248304), Guangdong Basic and Applied Basic Research Foundation (2023A1515011773).
\textbf{Author Contributions:}
H.S. developed the theorems, implemented the method, plotted experimental results, and wrote the manuscript.
X.W. supervised the research.
X.W. and C.X. assisted the research with constructive discussions.
C.X., J.T., and B.Y. assisted with manuscript editing.
All authors read and commented the paper.
\textbf{Competing interests:}
The authors declare no competing interests.
\textbf{Code availability:}
The codes implementing our method and animating line search process that takes any dynamic game instance as input are available at \url{https://github.com/shb20tsinghua/PTAS_Game/tree/main}.

\bibliographystyle{unsrt}

\section*{Proofs}
In this section, we provide the proofs of all the theorems and propositions in the previous sections.
\begin{proof}[Proof of Proposition~\ref{single_dp_lp}]
    (i) The dual problem of linear programming problem \eqref{primal_lp} is
    \begin{align*}
        \max_{\lambda_a^s}\quad & \lambda_a^s u_a^s                                \\
        \textrm{s.t.}\quad      & (I_{s'}^s-\gamma T_{s'a}^s)\lambda_a^s=\bar{w}^s \\
                                & \lambda_a^s\geq 0
    \end{align*}
    Note that $u_\pi^s=(I_{s'}^s-\gamma T_{\pi s'}^s)V_{\pi s'}$, then substituting in $\pi_a^s=\lambda_a^s/(\mathbf{1}_a\lambda_a^s)$, we have
    $$\lambda_a^s u_a^s=(\mathbf{1}_a\lambda_a^s)(I_{s'}^s-\gamma T_{\pi s'}^s)V_{\pi s'}=(\mathbf{1}_a\lambda_a^s)\pi_a^s(I_{s'}^s-\gamma T_{s'a}^s)V_{\pi s'}=\bar{w}^sV_{\pi s}.$$
    Thus, we obtain linear programming problem \eqref{dual_lp}.

    The optimal value of linear programming problem \eqref{primal_lp} being the optimal value function is an existing result,
    and the optimal value of linear programming problem \eqref{dual_lp} being the optimal value function is intuitive.
    In addition, the optimal values of these two linear programming problems are the same since they are the dual linear programming problem of each other.

    (ii) It directly follows.
\end{proof}

\subsection*{Proofs in section \ref{ipm_sec}}
\begin{proof}[Proof of Theorem~\ref{regmin_equiv}]
    First, use (i) to prove (ii) and the optimal objective function value is $0$.
    It follows from $(\pi_a^i,r_a^i) \geq 0$ that the objective $\pi_a^i r_a^i\geq 0$.
    When $(\pi_a^i,v^i)$ is a Nash equilibrium, $v^i=\pi_A U_A^i=\max_a\pi_{Aa}^{i-}U_A^i$.
    Then
    $$\pi_a^i r_a^i=\pi_a^i\left(v^i-\pi_{Aa}^{i-}U_A^i\right)=v^i-\pi_A U_A^i=0$$
    and $r_a^i=v^i-\pi_{Aa}^{i-}U_A^i\geq 0$.
    Hence $(\pi_a^i,r_a^i,v^i)$ is a global optimal point and the optimal objective function value is $0$.

    Then use (ii) to prove (iii).
    By the guaranteed existence of Nash equilibria and the above inference, the optimal objective function value is always $0$.
    It follows from $\pi_a^i r_a^i= 0$ and $(\pi_a^i,r_a^i) \geq 0$ that
    $$\pi_a^i\circ r_a^i=0.$$
    Hence let $(\bar{\lambda}_a^i,\tilde{\lambda}^i,\hat{\pi}_a^i,\hat{r}_a^i)=(\mathbf{0}_a^i,\mathbf{0}^i,\pi_a^i,r_a^i)$, and then the equations are satisfied.

    Finally, use (iii) to prove $(\bar{\lambda}_a^i,\tilde{\lambda}^i,\pi_a^i-\hat{\pi}_a^i,r_a^i-\hat{r}_a^i)=0$ and (i).
    Substituting $\pi_a^i=\hat{\pi}_a^i$ into $\bar{\lambda}_a^i+\pi_a^i-\hat{\pi}_a^i=0$, we have $\bar{\lambda}_a^i=0$.
    Substituting $\bar{\lambda}_a^i=0$ into $\bar{\lambda}_a^j\pi_{Aaa'}^{ij-}U_A^i+\tilde{\lambda}^i\mathbf{1}_{a'}+r_{a'}^i-\hat{r}_{a'}^i=0$, we have $r_a^i-\hat{r}_a^i=-\tilde{\lambda}^i\mathbf{1}_a$.
    Substituting $\pi_a^i=\hat{\pi}_a^i$ into $r_a^i\circ \hat{\pi}_a^i=0$ and $\pi_a^i\circ \hat{r}_a^i=0$, we have $\pi_a^i\circ(r_a^i-\hat{r}_a^i)=0$.
    Then
    $$\pi_a^i\circ\left(-\tilde{\lambda}^i\mathbf{1}_a\right)=0.$$

    Because $\mathbf{1}_a\pi_a^i=\mathbf{1}^i$, so for every index $i$ there must exist index $a$ such that $\pi_a^i>0$.
    It follows that $\tilde{\lambda}^i=0$ for every index $i$, and then $r_a^i=\hat{r}_a^i$.
    Hence we obtain $(\bar{\lambda}_a^i,\tilde{\lambda}^i,\pi_a^i-\hat{\pi}_a^i,r_a^i-\hat{r}_a^i)=0$.

    At this time, $\pi_a^i\circ r_a^i=0$.
    Considering that for every index $i$ there must exist index $a$ such that $\pi_a^i>0$,
    then for every index $i$ there must exist index $a$ such that $r_a^i=v^i-\pi_{Aa}^{i-}U_A^i=0$.
    Note also that $r_a^i \geq 0$, and thus
    $$v^i=\max_a\pi_{Aa}^{i-}U_A^i.$$
    Then it follows from the objective $\pi_a^i r_a^i= 0$ that $v^i=\max_a\pi_{Aa}^{i-}U_A^i=\pi_A U_A^i$.
    Hence $(\pi_a^i,v^i)$ is a Nash equilibrium.
\end{proof}

\begin{proof}[Proof of Theorem~\ref{ipm_equiv}]
    First, $(i)\Leftrightarrow(iii)$ is implied by the definition of Brouwer function $M$.

    Then, prove $(ii)\Leftrightarrow(iii)$.
    In unbiased barrier problem \eqref{ubarr_equ}, where $\hat{\pi}_a^i=\mu_a^i/r_a^i$ and $\hat{r}_a^i=\mu_a^i/\pi_a^i$,
    $$\left(\pi_a^i-\hat{\pi}_a^i\right)\left(r_a^i-\hat{r}_a^i\right)=\sum_{i,a}\left(\pi_a^i\circ r_a^i+\frac{\mu_a^{i2}}{\pi_a^i\circ r_a^i}-2\mu_a^i\right).$$
    The formula takes the minimum value if and only if $\pi_a^i\circ r_a^i=\mu_a^i$.
    Considering the constraints $r_a^i-v^i+\pi_{Aa}^{i-}U_A^i=0$ and $\mathbf{1}_a\hat{\pi}_a^i-\mathbf{1}^i=0$, the simultaneous equations are exactly unbiased KKT conditions \eqref{ukkt_equ}.
    Hence we obtain the equivalence.

    Then, prove $(iii)\Leftrightarrow((iv)\wedge \pi_a^i=\hat{\pi}_a^i)$.
    Unbiased KKT conditions \eqref{ukkt_equ} are exactly the simultaneous equations of perturbed KKT conditions \eqref{perturbed_kkt} and unbiased condition $\pi_a^i=\hat{\pi}_a^i$.
    Hence we obtain the equivalence.

    Then, prove $(iv)\Leftrightarrow(v)\Leftrightarrow(vi)$.
    It can be verified that KKT conditions of barrier problem \eqref{regmin_equ} are perturbed KKT conditions \eqref{perturbed_kkt},
    and simultaneous equations of KKT conditions and parameter $\hat{\pi}_a^i=\mu_a^i/r_a^i$ and $\hat{r}_a^i=\mu_a^i/\pi_a^i$ of unbiased barrier problem \eqref{ubarr_equ} are also perturbed KKT conditions \eqref{perturbed_kkt}.
    Note that unbiased barrier problem \eqref{ubarr_equ} and barrier problem \eqref{regmin_equ} are both equality constrained optimization problems,
    and thus by the Lagrange multiplier method, their local extreme points are the points that satisfies their KKT conditions, that is, perturbed KKT conditions \eqref{perturbed_kkt}.
    Hence we obtain the equivalence.

    Finally, $(\pi_a^i,\mathbf{0}_a^i)$ is a solution of unbiased KKT conditions \eqref{ukkt_equ} is equivalent to Theorem~\ref{regmin_equiv} (iii),
    and thus it is equivalent to $\pi_a^i$ being a Nash equilibrium.
\end{proof}

\begin{proof}[Proof of Theorem~\ref{ipm_theo}]
    (i) For every index $i$, denote the function and its derivative
    $$f(v^i)=\mathbf{1}_a\frac{\mu_a^i}{v^i-\pi_{Aa}^{i-}U_A^i}-\mathbf{1}^i,\frac{df(v^i)}{dv^i}=-\mathbf{1}_a\frac{\mu_a^i}{\left(v^i-\pi_{Aa}^{i-}U_A^i\right)^2}\leq 0.$$
    There is also
    $$\lim_{v^i\to \left(\max_a\pi_{Aa}^{i-}U_A^i\right)^+}f(v^i)=+\infty \quad \wedge \quad \lim_{v^i\to +\infty}f(v^i)=-1.$$
    Thus, $f(v^i)$ monotonically decreases with respect to $v^i$ from $+\infty$ to $-1$ in its domain, and hence there exists a unique $v^i$ that satisfies $f(v^i)=0$.

    (ii) By the objective function and constraint $r_a^i-v^i+\pi_{Aa}^{i-}U_A^i=0$, considering $\hat{\pi}_a^i$ and $\hat{r}_a^i$ are parameters, the differential of objective function is
    $$\left(\pi_a^i-\hat{\pi}_a^i\right) dr_a^i+\left(r_a^i-\hat{r}_a^i\right) d\pi_a^i=\left(\pi_a^i-\hat{\pi}_a^i\right) \left(dv^i-\pi_{Aaa'}^{ij-}U_A^i d\pi_{a'}^j\right) +\left(r_a^i-\hat{r}_a^i\right) d\pi_a^i.$$
    It follows from $\mathbf{1}_a\hat{\pi}_a^i-\mathbf{1}^i=0$ that $(\pi_a^i-\hat{\pi}_a^i)dv^i=0$, and thus the differential is
    $$\left(\left(r_{a'}^j-\hat{r}_{a'}^j\right)-\left(\pi_a^i-\hat{\pi}_a^i\right)\pi_{Aaa'}^{ij-}U_A^i\right)d\pi_{a'}^j=\left(\pi_a^i-\hat{\pi}_a^i\right)\left({\rm Diag}(r_a^i)-\pi_{Aaa'}^{ij-}U_A^i\circ\pi_{a'}^j\right)(d\pi_{a'}^j/\pi_{a'}^j).$$

    (iii) (i) implies that Brouwer function $\hat{\pi}_a^i=M(\mu_a^i)(\pi_a^i)$ is indeed a map, such that for every $\pi_a^i$ there is a unique $\hat{\pi}_a^i$.
    Note also that $\hat{\pi}_a^i=M(\mu_a^i)(\pi_a^i)$ is continuous, since equation \eqref{ukkt_equ1} is continuous.
    Thus, by Brouwer's fixed point theorem, there is always a fix point that satisfies $\hat{\pi}_a^i=\pi_a^i$, that is, a point that satisfies unbiased KKT conditions \eqref{ukkt_equ}.
\end{proof}

\begin{proof}[Proof of Theorem~\ref{equil_bund_equiv}]
    (i) and (ii) directly follow from definition of the equilibrium bundle and unbiased KKT conditions \eqref{ukkt_equ}.

    (iii) By Theorem~\ref{ipm_equiv}, $\pi_a^i$ is a Nash equilibrium if and only if $\mathbf{0}_a^i\in B(\pi_a^i)$, then the equivalence is directly implied.

\end{proof}

\begin{proof}[Proof of Theorem~\ref{equil_bund_theo}]
    (i) From unbiased KKT conditions \eqref{ukkt_equ} we have
    $$\pi_a^i\circ v^i-\pi_a^i\circ\pi_{Aa}^{i-}U_A^i=\hat{\mu}_a^i\mu'$$
    and $\mathbf{1}_a\pi_a^i=\mathbf{1}^i$.
    As $\mu' \to +\infty$, $\pi_a^i$ and $\pi_a^i\circ\pi_{Aa}^{i-}U_A^i$ is bounded, and $v^i\to +\infty$.
    Then
    $$\lim_{\mu' \to +\infty} \pi_a^i\circ \frac{v^i}{\mu'}=\hat{\mu}_a^i.$$
    Hence we obtain $\lim_{\mu' \to +\infty} \pi_a^i=\hat{\mu}_a^i/(\mathbf{1}_a\hat{\mu}_a^i)$.

    (ii) Differential of unbiased KKT conditions \eqref{ukkt_equ} is
    \begin{align*}
        \begin{bmatrix}
            \pi_a^i\circ dr_a^i+d\pi_a^i\circ r_a^i-d\mu_a^i \\
            dr_a^i-dv^i+\pi_{Aaa'}^{ij-}U_A^i d\pi_{a'}^j    \\
            \mathbf{1}_a d\pi_a^i
        \end{bmatrix}=0.
    \end{align*}
    Eliminating $dr_a^i$, we have
    \begin{align*}
        \begin{bmatrix}
            \pi_a^i\circ dv^i-\left(\pi_a^i\circ\pi_{Aaa'}^{ij-}U_A^i\circ\pi_{a'}^j\right)  \frac{d\pi_{a'}^j}{\pi_{a'}^j}+\mu_a^i\circ \frac{d\pi_a^i}{\pi_a^i} \\
            \mathbf{1}_a(\pi_a^i\circ \frac{d\pi_a^i}{\pi_a^i})
        \end{bmatrix}=\begin{bmatrix}d\mu_a^i\\0\end{bmatrix}.
    \end{align*}
    Transform it into a linear equation system with respect to $(d\pi_{a'}^j/\pi_{a'}^j)/(d\mu_{a''}^k/\mu_{a''}^k)$ and $dv^l/(d\mu_{a''}^k/\mu_{a''}^k)$, and we obtain equation \eqref{tan_vec}.

    (iii) The equation is obtained from the proof of Theorem~\ref{ipm_theo} (ii).
    $(\pi_a^i,\mu_a^i)\in E$ is a solution of unbiased KKT conditions \eqref{ukkt_equ}, then by Theorem~\ref{ipm_equiv} (v), $(\pi_a^i,\mu_a^i)$ is a KKT point of unbiased barrier problem \eqref{ubarr_equ}.

    First, when ${\rm C}_{(j,a')\cup l}^{(i,a)\cup m}$ is non-singular, the Jacobian matrix of unbiased KKT conditions \eqref{ukkt_equ} is non-singular, and then by the implicit function theorem, $(\pi_a^i,\mu_a^i)$ is the only solution of \eqref{ukkt_equ} in its neighborhood.
    Thus, $(\pi_a^i,\mu_a^i)$ is the only KKT point of \eqref{ubarr_equ}, namely, unbiased barrier problem \eqref{ubarr_equ} is locally strictly convex at $(\pi_a^i,\mu_a^i)$.

    Conversely, when ${\rm C}_{(j,a')\cup l}^{(i,a)\cup m}$ is singular, the KKT point of \eqref{ubarr_equ} is not unique in the neighborhood, namely, unbiased barrier problem \eqref{ubarr_equ} is not locally strictly convex at $(\pi_a^i,\mu_a^i)$.

    (iv) On the fiber $B(\pi_a^i)$ over $\pi_a^i$, $\mu_a^i/\pi_a^i$ tends to $v^i\mathbf{1}_a$ for some $v^i$ as $\mu_a^i\to \infty$, and $\pi_{Aaa'}^{ij-}U_A^i\circ\pi_{a'}^j$ is bounded.
    Thus, matrix ${\rm C}_{(j,a')\cup l}^{(i,a)\cup m}$ tends to
    \begin{align*}
        \begin{bmatrix}
            {\rm Diag}(v^i\mathbf{1}_a)                                      & \left(I^{il}\mathbf{1}_a\right) _l^{(i,a)} \\
            \pi_{(j,a')}\circ \left(I^{jm}\mathbf{1}_{a'}\right) _{(j,a')}^m & \mathbf{0}_l^m
        \end{bmatrix}.
    \end{align*}
    By elementary column transformation, the determinant of this matrix is $(-1)^{\left\lvert N\right\rvert }\prod_{i\in N} (v^i)^{\left\lvert \mathcal{A}\right\rvert -1}$.
    It follows that the determinant of ${\rm C}_{(j,a')\cup l}^{(i,a)\cup m}$ diverges as $\mu_a^i\to \infty$ in $B(\pi_a^i)$.
    Thus, there always exists $\check{\mu}_a^i\in B(\pi_a^i)$ such that for every $\mu_a^i>\check{\mu}_a^i$, the determinant of ${\rm C}_{(j,a')\cup l}^{(i,a)\cup m}$ is non-zero on $(\pi_a^i,\mu_a^i)$.

\end{proof}

\begin{proof}[Proof of Theorem~\ref{odd_thm}]
    (i) First, unbiased KKT conditions \eqref{ukkt_equ} derive a polynomial function that maps $(\pi_a^i,v^i)$ to $\mu_a^i$, and polynomial functions are analytic on their whole domain.
    Second, function with respect to $\mu_a^i$ is given by a polynomial equation, such that the function expands to a Puiseux's series with respect to $\mu_a^i$ at every point according to the Newton-Puiseux theorem.
    Then, this function with respect to $\mu_a^i$ is analytic in the punctured neighbourhood of every point on it.
    Furthermore, this function is analytic in the neighbourhood of every point on it, because the Puiseux's series also converges at the point itself since the function is the inverse of a polynomial function.

    (ii) For any $\hat{\mu}_a^i>0$ and $\mu''\geq 0$, there is an algebraic curve ${\rm AC}=\{(\pi_a^i,\hat{\mu}_a^i\mu')\in E|\mu'> \mu''\}$.
    Denote ${\rm EP}$ as the set points on ${\rm AC}$ as $\mu'\to \mu''$.
    By (i), at every point on an algebraic curve, either non-singular or singular, there is a unique analytic continuation beyond that point.
    It follows that the endpoints of an algebraic curve are always connected in pairs.
    And Theorem~\ref{equil_bund_theo} (i) shows that there is a unique endpoint of algebraic curve ${\rm AC}$ as $\mu'\to \infty$, which is called the starting point.
    Thus, there is exactly one point in ${\rm EP}$ that is connected to the starting point of ${\rm AC}$ by a branch of ${\rm AC}$, and all the other points in ${\rm EP}$ are connected in pairs by the rest branches of ${\rm AC}$.
    In other words, there are always an odd number of points in ${\rm EP}$.

    Then we show that ${\rm EP}$ almost always equals to $\{(\pi_a^i,\hat{\mu}_a^i\mu'')\in E\}$.
    For a point $(\pi_a^i,\mu_a^i)$ where ${\rm C}_{(j,a')\cup l}^{(i,a)\cup m}$ is non-singular,
    $\pi_a^i$ is the unique solution in its neighborhood for the given $\mu_a^i$ by the implicit function theorem.
    It follows that if $(\pi_a^i,\hat{\mu}_a^i\mu'')\in {\rm EP}$ is a non-singular point,
    then for every $\hat{\mu}_a^i>0$, the algebraic curve $\{(\pi_a^i,\hat{\mu}_a^i\mu'+\hat{\mu}_a^i\mu'')\in E|\mu'> 0\}$ tends to $(\pi_a^i,\hat{\mu}_a^i\mu'')$ as $\mu'\to 0$.
    In other words, $\{(\pi_a^i,\hat{\mu}_a^i\mu'')\in E\}$ is contained in ${\rm EP}$ if every point in ${\rm EP}$ is non-singular, and thus they are equal if so.
    Note that almost all points given by a polynomial equation are non-singular.
    Thus, ${\rm EP}$ almost always equals to $\{(\pi_a^i,\hat{\mu}_a^i\mu'')\in E\}$.

\end{proof}

\subsection*{Proofs in section \ref{dp_sec}}

\begin{lemma}
    \label{stoc_mat_lemm}
    Let $T_\pi\in \mathbb{R}^{n\times n}$ be a matrix that satisfies $T_\pi\geq 0$ and $T_\pi\mathbf{1}=\mathbf{1} $, and let $\gamma \in [0,1)$.
    Then
    \begin{enumerate}
        \item $(I-\gamma T_\pi)$ is invertible.
        \item For any $X\in \mathbb{R}^n$, the formula $(I-\gamma T_\pi)X\geq 0\rightarrow X\geq 0\rightarrow (\gamma T_\pi)X\geq 0$ holds.
    \end{enumerate}
\end{lemma}
\begin{proof}
    (i) The eigenvalues of $(I-\gamma T_\pi)$ is given by $1-\gamma\lambda_i$, where $\lambda_i$ are the eigenvalues of $T_\pi$.
    Note that inequality $T_\pi X\leq \max(X)$ holds, and thus for real eigenvalue $\lambda_i$ that $T_\pi X=\lambda_i X$, we have $\lambda_i\leq 1$, and then $1-\gamma\lambda_i>0$.
    Hence $(I-\gamma T_\pi)$ has no zero eigenvalues, and it's invertible.

    (ii) From $T_\pi\geq 0$ and $\gamma \in [0,1)$, we obtain $X\geq 0\rightarrow (\gamma T_\pi)X\geq 0$.
    Note that inequality $T_\pi X\geq \min(X)$ holds, and it follows that
    $$\min((I-\gamma T_\pi)X)\leq \min(X)-\gamma \min(T_\pi X)\leq (1-\gamma)\min(X).$$
    Hence we obtain $(I-\gamma T_\pi)X\geq 0\rightarrow X\geq 0$.
\end{proof}

\begin{proof}[Proof of Proposition~\ref{cone_prop}]
    (i) By Lemma~\ref{stoc_mat_lemm} (i), for any policy $\pi_a^{si}$, there is a unique value function $$V_{\pi s}^i=(I_{s'}^s-\gamma T_{\pi s'}^s)^{-1}u_\pi^{s'i}$$ that satisfies $V_{\pi s}^i=\pi_A^s(u_A^{si}+\gamma T_{s'A}^s V_{\pi s'}^i)$.
    Then by definition of $C_\pi$, we obtain $V_{\pi s}^i\in C_\pi$.
    It follows from $V_s^i \in C_\pi$ that
    $$\left(I_{s'}^s-\gamma T_{\pi s'}^s\right)\left(V_s^i-V_{\pi s}^i\right)=V_s^i-\pi_A^s\left(u_A^{si}+\gamma T_{s'A}^s V_{s'}^i\right)\geq 0.$$
    Then by Lemma~\ref{stoc_mat_lemm} (ii), we have $V_s^i-V_{\pi s}^i\geq 0$, and thus for all $V_s^i \in C_\pi$, $V_s^i\geq V_{\pi s}^i$.

    (ii) Considering $V_s^i-Y_{xs}^i=d_x^i\mathbf{1}_s$, there is
    $$\left(Y_{xs}^i-D_\pi(Y_{xs}^i)\right)(x)=\left(V_s^i-D_\pi(V_s^i)-(1-\gamma)d_x^i\mathbf{1}_s\right)(x).$$
    Then let $d_x^i=(V_s^i-D_\pi(V_s^i))/(1-\gamma)$, and thus formula \eqref{cone_boundary} is satisfied for every $x \in \mathcal{S} $.
    Note that $d_x^i$ is unique, and hence $Y_{xs}^i$ is unique.

    Similarly, considering $V_s^i-\hat{Y}_{xs}^i=\hat{d}_x^i\mathbf{1}_s$, there is
    $$\left(\hat{Y}_{xs}^i-\hat{D}_\pi(\hat{Y}_{xs}^i)\right)(x)=\left(V_s^i-\hat{D}_\pi(V_s^i)-(1-\gamma)\hat{d}_x^i\mathbf{1}_s\right)(x).$$
    then let $\hat{d}_x^i=(V_s^i-\hat{D}_\pi(V_s^i))/(1-\gamma)$, and thus formula \eqref{coneb_boundary} is satisfied for every $x \in \mathcal{S} $.
    Note that $\hat{d}_x^i$ is unique, and hence $\hat{Y}_{xs}^i$ is unique.

    (iii) For any $V_s^i \in \mathcal{V} $, we have
    $$V_s^i+m^i\mathbf{1}_s-\pi_{Aa}^{si-}\left(u_A^{si}+\gamma T_{s'A}^s \left(V_{s'}^i+m^i\mathbf{1}_{s'}\right)\right)=V_s^i-\pi_{Aa}^{si-}\left(u_A^{si}+\gamma T_{s'A}^s V_{s'}^i\right)+(1-\gamma)m^i\mathbf{1}_s,$$
    where $V_s^i-\pi_{Aa}^{si-}(u_A^{si}+\gamma T_{s'A}^s V_{s'}^i)$ is constant.
    Then there exists an $M^i>0$ such that for any $m^i>M^i$ the above formula is greater than 0,
    and hence $V_s^i+m^i\mathbf{1}_s \in \hat{C}_\pi$.

    By definition, for any $V_s^i \in \hat{C}_\pi$, we have
    $$V_s^i\geq \pi_{Aa}^{si-}\left(u_A^{si}+\gamma T_{s'A}^s V_{s'}^i\right)\geq\pi_A^s\left(u_A^{si}+\gamma T_{s'A}^s V_{s'}^i\right).$$
    Then $V_s^i \in C_\pi$, and hence $\hat{C}_\pi \subseteq C_\pi$.
\end{proof}

\begin{proof}[Proof of Theorem~\ref{cone_dp}]
    (i) It follows directly from definition of $C_\pi$.

    (ii) Note that equation
    \begin{align*}
        D_\pi(V_s^i)-D_\pi\left(D_\pi(V_s^i)\right) =\left(I_{s'}^s-\gamma T_{\pi s'}^s\right)\left(V_s^i-V_{\pi s}^i-\left(V_s^i-D_\pi(V_s^i)\right)\right) =\gamma T_{\pi s'}^s\left(V_s^i-D_\pi(V_s^i)\right)
    \end{align*}
    holds.
    Then if $V_s^i \in C_\pi$, by Lemma~\ref{stoc_mat_lemm} (ii), we obtain $D_\pi(V_s^i) \in C_\pi$.

    (iii) Note that equation
    $$D_\pi(V_s^i)-\hat{D}_\pi\left(D_\pi(V_s^i)\right)=\gamma \tilde{\pi}_a^{si}\pi_{Aa}^{si-}T_{s'A}^s\left(V_s^i-D_\pi(V_s^i)\right)+D_\pi(V_s^i)-\hat{D}_\pi(V_s^i)$$
    holds.
    By Lemma~\ref{stoc_mat_lemm} (ii), if $D_\pi(V_s^i)= \hat{D}_\pi(V_s^i)$, then for any $V_s^i \in C_\pi$, $D_\pi(V_s^i)\geq \hat{D}_\pi(D_\pi(V_s^i))$, and hence $D_\pi(V_s^i) \in \hat{C}_\pi$.

    (iv) It follows directly from the proof of Proposition~\ref{cone_prop} (ii).

    (v) The iteration formula is equivalent to $\tilde{V}_{s,k+1}^i=D_\pi(\tilde{V}_{s,k}^i)$, where $\tilde{V}_{s,k}^i=V_{s,k}^i-\gamma m^i\mathbf{1}_s/(1-\gamma)$.
    It follows from $V_{s,0}^i\in C_\pi$ that $V'_0\in C(\pi)$ by (ii), and thus $\tilde{V}_{s,k}^i\geq \tilde{V}_{s,k+1}^i$ and $\tilde{V}_{s,k}^i\in C_\pi$ for all $k\in\mathbb{N}$ by (i).
    According to the monotone convergence theorem, $\tilde{V}_{s,k}^i$ converges as $k \to \infty$, and the limit is the unique solution $V_{\pi s}^i$ of $V_{\pi s}^i=D_\pi(V_{\pi s}^i)$.
    Hence we obtain
    $$\lim_{k \to \infty} V_{s,k}^i=V_{\pi s}^i+\frac{\gamma}{1-\gamma}m^i\mathbf{1}_s \quad \wedge \quad \lim_{k \to \infty} \left(V_{s,k}^i-D_\pi(V_{s,k}^i)\right)=\gamma m^i\mathbf{1}_s.$$
\end{proof}

\begin{proof}[Proof of Theorem~\ref{cone_equil}]
    (i) By definition, $\pi_a^{si}$ is a perfect equilibrium if and only if
    $$V_{\pi s}^i=\max_a\pi_{Aa}^{si-}\left(u_A^{si}+\gamma T_{s'A}^s V_{\pi s'}^i\right),$$
    and $V_{\pi s}^i \in \hat{C}_\pi$ if and only if
    $$V_{\pi s}^i\geq \pi_{Aa}^{si-}\left(u_A^{si}+\gamma T_{s'A}^s V_{\pi s'}^i\right).$$
    Note that the equality can be established in the above inequality, because $V_{\pi s}^i=\pi_A^s\left(u_A^{si}+\gamma T_{s'A}^s V_{\pi s'}^i\right)$.
    Then the two formulas are equivalent, and hence $\pi_a^{si}$ is a perfect equilibrium if and only if $V_{\pi s}^i \in \hat{C}_\pi$.

    It directly follows from the definitions that $V_{\pi s}^i \in \hat{C}_\pi$ if and only if $V_{\pi s}^i$ is the global optimal point of dynamic programming problem \eqref{dp_problem}.

    (ii) For every $x \in \mathcal{S} $, $\pi_a^{si}(x)$ is a Nash equilibrium if and only if
    $$D_\pi(V_s^i)(x)=\hat{D}_\pi(V_s^i)(x).$$

    First, suppose $\pi_a^{si}(x)$ is a Nash equilibrium.
    Consider $Y_{xs}^i$ and $d_x^i$ that satisfies formula \eqref{cone_boundary},
    and then substituting with $V_s^i=Y_{xs}^i+d_x^i\mathbf{1}_s$, we have
    $$D_\pi(Y_{xs}^i)(x)=\hat{D}_\pi(Y_{xs}^i)(x),$$
    and further there is
    $$\left(Y_{xs}^i-\hat{D}_\pi(Y_{xs}^i)\right)(x)=0.$$
    By the uniqueness of the pair of $\hat{Y}_{xs}^i$ and $\hat{d}_x^i$, we obtain $Y_{xs}^i=\hat{Y}_{xs}^i$.

    Conversely, suppose $Y_{xs}^i=\hat{Y}_{xs}^i$, and it follows that $d_x^i=\hat{d}_x^i$.
    By formula \eqref{cone_boundary} and \eqref{coneb_boundary} there are
    \begin{align*}
        \left(V_s^i-D_\pi(V_s^i)-(1-\gamma)d_x^i\mathbf{1}_s\right)(x)=             & 0, \\
        \left(V_s^i-\hat{D}_\pi(V_s^i)-(1-\gamma)\hat{d}_x^i\mathbf{1}_s\right)(x)= & 0.
    \end{align*}
    Then $D_\pi(V_s^i)(x)=\hat{D}_\pi(V_s^i)(x)$, and hence $\pi_a^{si}(x)$ is a Nash equilibrium.
\end{proof}

\begin{proof}[Proof of Proposition~\ref{naive_iter}]
    Note that the limit is always a perfect equilibrium as long as $V_{s,k}^i$ converges under the assumption that $D_{\pi_k}(V_s^i)=\hat{D}_{\pi_k}(V_s^i)$,
    so it is suffice to show that $V_{s,k}^i$ converges.

    First, use (i) to prove (ii).
    Consider $V_s^i-\delta\leq Y_s^i \leq V_s^i+\delta$, where $V_s^i,Y_s^i\in O(\hat{V}_s^i)$ and $\delta=\lVert V_s^i-Y_s^i\rVert _\infty$ is small enough.
    Using $V_s^i\leq Y_s^i\rightarrow D(V_s^i)\leq D(Y_s^i)$, there is
    $$D(V_s^i)-\gamma\delta\leq D(Y_s^i) \leq D(V_s^i)+\gamma\delta.$$
    Then we obtain $\lVert D(V_s^i)-D(Y_s^i)\rVert _\infty\leq\gamma\lVert V_s^i-Y_s^i\rVert _\infty$.
    Thus, $D$ is a contraction mapping on $O(\hat{V}_s^i)$, and $V_{s,k}^i\in O(\hat{V}_s^i)$ converges by the contraction mapping theorem.

    By $V_s^i\leq Y_s^i\rightarrow D(V_s^i)\leq D(Y_s^i)$ and $V_{s,0}^i\in C_{\pi_0}$,
    we have $V_{s,k+1}^i\leq V_{s,k}^i$ for all $k$, and hence we obtain $V_{s,k}^i\in C_{\pi_k}$ for all $k$.

    Then use (ii) to prove (iii).
    Monotonicity follows directly from $V_{s,k}^i\in C_{\pi_k}$ for all $k$, and convergence already holds, and hence (iii) is obtained.

    Then use (iii) to prove (iv).
    It is obtained directly from monotonicity.

    Finally, use (iv) to prove (iii).
    It follows from $V_{s,k+1}^i\leq V_{s,k}^i\rightarrow D(V_{s,k+1}^i)\leq D(V_{s,k}^i)$ that
    $$V_{s,k}^i\in C_{\pi_k}\rightarrow V_{s,k+1}^i \in C_{\pi_{k+1}},$$
    and thus $V_{s,k}^i\in C_{\pi_k}$ for all $k$.
    Considering $V_{s,0}^i\in C_{\pi_0}$, it follows that $V_{s,k}^i$ monotonically decreases by Theorem~\ref{cone_dp} (i).

    Note that $V_s^i<\pi_A^s(u_A^{si}+\gamma T_{s'A}^s V_{s'}^i)$ for any $\pi_a^{si}$ when $V_s^i<\min_A^{si} u_A^{si}/(1-\gamma)$,
    that is, $\bigcup_\pi C_\pi$ is bounded, and thus there is a lower bound for $V_{s,k}^i$.
    Hence $V_{s,k}^i$ converges by the monotone convergence theorem.
\end{proof}

\begin{proof}[Proof of Theorem~\ref{cone_iter}]
    The existence of a sequence $\{m_k^i\}_{k\in \mathbb{N} }$ such that $V_{s,k}^i+m_k^i\mathbf{1}_s \in \hat{C}_{\pi_k}$ for all $k\in \mathbb{N}$ is guaranteed by Proposition~\ref{cone_prop} (iii).

    First, suppose $V_{s,k}^i$ converges to a perfect equilibrium value function $\tilde{V}_s^i$, then $\tilde{V}_s^i=D_{\tilde{\pi}}(\tilde{V}_s^i)$ for some $\tilde{\pi}_a^{si}$, and thus $\lim_{k \to \infty}m_k^i=0$.

    Conversely, suppose $\lim_{k \to \infty}m_k^i=0$, then $V_{s,k}^i$ monotonically decreases when $k$ is large enough since $V_{s,k}^i+m_k^i\mathbf{1}_s \in \hat{C}_{\pi_k}$.
    Note that $V_{\pi s}^i$ is bounded below, and thus $V_{s,k}^i$ satisfies the monotone convergence property and converges to some $\tilde{V}_s^i$.
    It follows from $\lim_{k \to \infty}m_k^i=0$ that $\tilde{V}_s^i=V_{\tilde{\pi} s}^i$ and $\tilde{V}_s^i \in \hat{C}_{\tilde{\pi}}$ for some $\tilde{\pi}_a^{si}$.
    Hence, $\tilde{V}_s^i$ is a perfect equilibrium value function.
\end{proof}

\subsection*{Proofs in section \ref{hybrid_sec}}

\begin{proof}[Proof of Theorem~\ref{per_equil_thm}]
    Note that $U_A^{si}=u_A^{si}+\gamma T_{s'A}^s V_{s'}^i$, then there is
    \begin{align*}
          & \hat{D}_\pi(V_s^i)-D_\pi(V_s^i)=\max_a \pi_{Aa}^{si-}(u_A^{si}+\gamma T_{s'A}^s V_{s'}^i)-\pi_A^s(u_A^{si}+\gamma T_{s'A}^s V_{s'}^i)                  \\
        = & \sum_a \left(\pi_a^{si} \circ \left(\max_a \pi_{Aa}^{si-}U_A^{si}-\pi_{Aa}^{si-}U_A^{si}\right)\right) =\mathbf{1}_a\bar{\mu}_a^{si}(V_s^i,\pi_a^{si})
    \end{align*}

    (i) being equivalent to (ii) follows from the definition of perfect equilibrium that $V_s^i=V_{\pi s}^i$ and $\pi_a^{si}$ is a Nash equilibrium of $U_{\pi A}^{si}=u_A^{si}+\gamma T_{s'A}^s V_{\pi s'}^i$ for every state $s$,
    which is equivalent to $\bar{\mu}_a^{si}(V_{\pi s}^i,\pi_a^{si})=0$ by Theorem~\ref{equil_bund_equiv}~(iii).

    (i) being equivalent to (iii) follows from $V_{\pi s}^i=D_\pi(V_{\pi s}^i)$ and Theorem~\ref{cone_equil} (i).

    (i) being equivalent to (iv) follows from (iii) and that the objective $\pi_a^{si} r_a^{si}$ of problem \eqref{regmin_problem} is larger than $0$.
\end{proof}

\begin{proof}[Proof of Theorem~\ref{exist_odd_thm}]
    (i) Similar to Theorem~\ref{ipm_theo}, $\hat{\pi}_a^{si}=M(\mu_a^{si})(\pi_a^{si})$ is still a map, such that the Brouwer's fixed point theorem still applies.

    (ii) Similar to Theorem~\ref{odd_thm}, unbiased KKT conditions \eqref{ukkt_dy_equ} still derive a polynomial function from $\pi_a^{si}$ to $\mu_a^{si}$ by multiplying the determinant $|I_{s'}^s-\gamma T_{\pi s'}^s|$ to the second line,
    where $|I_{s'}^s-\gamma T_{\pi s'}^s|$ is non-zero for every $\pi_a^{si}\in\mathcal{P}$ due to Lemma~\ref{stoc_mat_lemm} (ii), such that the Newton-Puiseux theorem still applies.
\end{proof}

\begin{proof}[Proof of Proposition~\ref{conv_rate}]
    (i) By the iteration formula, the residual $V_s^i-D_\pi(V_s^i+m^i\mathbf{1}_s)$ satisfies
    $$D_\pi(V_s^i+m^i\mathbf{1}_s)-D_\pi(D_\pi(V_s^i+m^i\mathbf{1}_s)+m^i\mathbf{1}_s)=\gamma T_{\pi s'}^s\left(V_s^i-D_\pi(V_s^i+m^i\mathbf{1}_s)\right).$$
    Note that $\gamma\in[0,1)$ and $T_{\pi s'}^s$ are constant, and $V_s^i-D_\pi(V_s^i+m^i\mathbf{1}_s)$ converges to $0$ as iteration.
    Thus, $\lVert V_s^i-D_\pi(V_s^i+m^i\mathbf{1}_s)\rVert _\infty$ converge to $0$ linearly with a rate of $\gamma$.

    (ii) Denote $\tilde{\mu}_a^i(\pi_a^i):=\mu_a^i-(\mathbf{1}_a\mu_a^i)\circ\pi_a^i$, the difference of $\tilde{\mu}_a^i(\pi_a^i)$ is
    $$d\tilde{\mu}_a^i(\pi_a^i)=d\mu_a^i-(\mathbf{1}_a(\mu_a^i+d\mu_a^i))\circ d\pi_a^i-(\mathbf{1}_a d\mu_a^i)\circ\pi_a^i.$$
    According to the Taylor's formula of $d\pi_a^i$ with respect to $d\mu_a^i-(\mathbf{1}_a d\mu_a^i)\circ\pi_a^i$, there is
    $$d\pi_a^i= \pi_a^i\circ\left(\frac{\mu_{a'}^j d\pi_a^i}{\pi_a^i d\mu_{a'}^j}\frac{d\mu_{a'}^j-(\mathbf{1}_{a'} d\mu_{a'}^j)\circ\pi_{a'}^j}{\mu_{a'}^j}\right)+o(d\mu_a^i-(\mathbf{1}_a d\mu_a^i)\circ\pi_a^i).$$
    Then we substitute in $d\pi_a^i$ and $d\mu_a^i=-\eta^i\circ\mu_a^i+\beta^i\circ\pi_a^i$, while noting that the part of $d\pi_a^i$ due to the term $\beta^i\circ\pi_a^i$ is $0$ for any $\beta^i$ on the equilibrium bundle, and obtain
    \begin{align*}
        d\tilde{\mu}_a^i(\pi_a^i) & = -\eta^j\circ\left(I-(1-\eta^i)\circ\mu_a^i\circ\frac{\mu_{a'}^j d\pi_a^i}{\pi_a^i d\mu_{a'}^j}\circ{\mu_{a'}^j}^{-1}\right)\tilde{\mu}_{a'}^j(\pi_a^i)                                                                  \\
                                  & +\eta^j\circ(1-\eta^i)\circ\left((\mathbf{1}_a\mu_a^i)\circ o(\tilde{\mu}_a^i(\pi_a^i))-\tilde{\mu}_a^i(\pi_a^i)\circ\frac{\mu_{a'}^j d\pi_a^i}{\pi_a^i d\mu_{a'}^j}\frac{\tilde{\mu}_{a'}^j(\pi_a^i)}{\mu_{a'}^j}\right)
    \end{align*}
    The second additive is a higher-order infinitesimal with respect to $\tilde{\mu}_a^i(\pi_a^i)$.
    The differential $(\mu_{a'}^j d\pi_a^i)/(\pi_a^i d\mu_{a'}^j)$ of the equilibrium bundle tends to $I$ as $\mu_a^i\to\infty$ on the fiber over $\pi_a^i$, then there always exists a large enough $\mu_a^i$ on the fiber over $\pi_a^i$ such that the eigenvalues of $(I-(1-\eta^i)\circ\mu_a^i\circ((\mu_{a'}^j d\pi_a^i)/(\pi_a^i d\mu_{a'}^j))\circ{\mu_{a'}^j}^{-1})$ are all positive for a given $\eta^i$.
    Note that $\tilde{\mu}_a^i(\pi_a^i)$ converges to $0$ as long as $\pi_a^i$ converges according to the iteration formula, and $\mu_a^i\to 0$ is sufficient to let $\pi_a^i$ converge.
    Thus, $\lVert \tilde{\mu}_a^i(\pi_a^i)\rVert _\infty$ converges to $0$ linearly.

    Then we show $\lVert\bar{\mu}_a^i(\pi_a^i)\rVert _\infty$ also converges to $0$ linearly.
    $$\tilde{\mu}_a^i(\pi_a^i)=\left(\mu_a^i-(\mathbf{1}_a\mu_a^i)\circ\pi_a^i\right)=\pi_a^i \circ \left(\pi_A U_{\pi A}^i-\pi_{Aa}^{i-}U_{\pi A}^i\right)$$
    $$\bar{\mu}_a^i(\pi_a^i)=\pi_a^i \circ \left(\max_a \pi_{Aa}^{i-}U_{\pi A}^i-\pi_{Aa}^{i-}U_{\pi A}^i\right)=\tilde{\mu}_a^i(\pi_a^i)+\pi_a^i \circ \left(\max_a \pi_{Aa}^{i-}U_{\pi A}^i-\pi_A U_{\pi A}^i\right)$$
    Note that when $\lVert \tilde{\mu}_a^i(\pi_a^i)\rVert _\infty<\epsilon$, there exists a coefficient $k$ irrelevant to $\epsilon$ such that either $\pi_a^i<k\epsilon$ or $|\pi_A U_{\pi A}^i-\pi_{Aa}^{i-}U_{\pi A}^i|<k\epsilon$ for every index $(i,a)$.
    Then we have specifically $|\pi_A U_{\pi A}^i-\max_a \pi_{Aa}^{i-}U_{\pi A}^i|<k\epsilon$ for the corresponding index since the algorithm leads to Nash equilibria.
    It follows that $\lVert \tilde{\mu}_a^i(\pi_a^i)\rVert _\infty<\epsilon$ implies $\lVert\bar{\mu}_a^i(\pi_a^i)\rVert _\infty<(k+1)\epsilon$ with $k$ irrelevant to $\epsilon$.
    Hence, $\lVert\bar{\mu}_a^i(\pi_a^i)\rVert _\infty$ converges to $0$ linearly.

    (iii) Consider unbiased barrier problem \eqref{ubarr_equ} with $(\hat{\pi}_a^i,\hat{r}_a^i)$ being variables instead of parameters.
    By the objective function and constraint $r_a^i-v^i+\pi_{Aa}^{i-}U_A^i=0$, and noting that $\mathbf{1}_a\hat{\pi}_a^i-\mathbf{1}^i=0$ implies $(\pi_a^i-\hat{\pi}_a^i)dv^i=0$, the differential of objective function is
    \begin{align*}
          & \left(\pi_a^i-\hat{\pi}_a^i\right) \left(dr_a^i+\mu_a^i/{\pi_a^i}^2\circ d\pi_a^i\right) +\left(r_a^i-\hat{r}_a^i\right) \left(d\pi_a^i+\mu_a^i/{r_a^i}^2\circ dr_a^i\right)                                                                                                                \\
        = & \left(\left(\left(2-\frac{\pi_a^i-\hat{\pi}_a^i}{\pi_a^i}\right)\circ\left(\pi_a^i-\hat{\pi}_a^i\right)\right)\left({\rm Diag}(r_a^i)-\pi_{Aaa'}^{ij-}U_A^i\circ\pi_{a'}^j\right)+\left(\pi_{a'}^j-\hat{\pi}_{a'}^j\right)^2 \frac{dv^j}{d\pi_{a'}^j}\right) \frac{d\pi_{a'}^j}{\pi_{a'}^j} \\
    \end{align*}
    This projected gradient and the projected gradient \eqref{projgrad_formula} only differ by two higher-order infinitesimals.
    Note also that on a non-singular point, unbiased barrier problem \eqref{ubarr_equ} is locally strictly convex as Theorem~\ref{equil_bund_theo} (iii) has proved.
    Considering that projected gradient descent is known to converge sublinearly, projected gradient descent \eqref{projgrad_formula} converges sublinearly if the convergence point is non-singular.

\end{proof}

\end{document}